%% file: ISCA2025ACMtemplate/main.tex
\documentclass[conference]{IEEEtran}
\IEEEoverridecommandlockouts

\usepackage{cite}
\usepackage{amsmath,amssymb,amsfonts}
\usepackage{algorithmic}
\usepackage{graphicx}
\usepackage{textcomp}
\usepackage{xcolor}
\usepackage{multirow}
\usepackage{threeparttable}
\usepackage[hyphens]{url}
\usepackage{graphicx}
\usepackage{enumitem}
\usepackage{subcaption}    
\usepackage{booktabs}
\usepackage{siunitx}
\usepackage{makecell}
\usepackage{xspace}
\usepackage{adjustbox}
\usepackage{pifont}
\usepackage{amsthm}
\newtheorem{theorem}{Theorem}

\usepackage[most]{tcolorbox}

\usepackage{listings}

\newcommand{\diff}[1]{{\color{black} #1}}

\newcommand{\sys}{\textsc{Lippen}\xspace}

\lstdefinestyle{CStyle}{
  language=C,
  basicstyle=\ttfamily\small,
  keywordstyle=\color{blue}\bfseries,
  numbers=left, numberstyle=\tiny\color{gray},
  numbersep=-30pt, 
  xleftmargin=-32pt,  
  framexleftmargin=0pt,
  columns=fullflexible, keepspaces=true,
  aboveskip=0pt, belowskip=0pt
}

\def\BibTeX{{\rm B\kern-.05em{\sc i\kern-.025em b}\kern-.08em
    T\kern-.1667em\lower.7ex\hbox{E}\kern-.125emX}}
\begin{document}

\pdfpagewidth=8.5in
\pdfpageheight=11in

\newcommand{\iscasubmissionnumber}{NaN}

\pagenumbering{arabic}

\title{\sys: A Lightweight In-Place Pointer Encryption Architecture for Pointer Integrity}


\author{%
\begin{tabular}{cccc}
Erfan Iravani &
Lalit Prasad Peri &
Mohannad Ismail &
Charitha Tumkur Siddalingaradhya \\
\textit{Virginia Tech} &
\textit{Virginia Tech} &
\textit{Virginia Tech} &
\textit{Virginia Tech} \\
{\small erfani@vt.edu} &
{\small lalitprasad@vt.edu} &
{\small imohannad@vt.edu} &
{\small charitha24@vt.edu}
\\[1em]
\multicolumn{1}{c}{Changwoo Min} &
\multicolumn{2}{c}{Elif Bilge Kavun} &
\multicolumn{1}{c}{Wenjie Xiong} \\
\multicolumn{1}{c}{\textit{Igalia}} &
\multicolumn{2}{c}{\textit{Barkhausen Institut \& TU Dresden}} &
\multicolumn{1}{c}{\textit{Virginia Tech}} \\
\multicolumn{1}{c}{{\small changwoo@igalia.com}} &
\multicolumn{2}{c}{{\small elif.kavun@barkhauseninstitut.org}} &
\multicolumn{1}{c}{{\small wenjiex@vt.edu}}
\end{tabular}%
}
\maketitle
\thispagestyle{plain}
\pagestyle{plain}


\begin{abstract}
Memory-safety violations in C and C++ programs continue to enable sophisticated exploitation techniques such as control-flow hijacking and data-oriented attacks. Existing hardware defenses either rely on address space layout randomization (ASLR) or attach explicit metadata to pointers to verify their integrity. External metadata schemes provide strong guarantees, but incur additional memory accesses and memory footprint overhead. In-place authentication mechanisms, such as ARM Pointer Authentication (PAC), achieve low overhead at the cost of limited entropy and susceptibility to brute-force and reuse attacks. 
This paper presents \sys, a hardware–software co-design for \emph{full-pointer encryption} that provides strong pointer integrity and confidentiality with zero metadata overhead. \sys treats every pointer as an encrypted block, cryptographically binding it to its execution context and decrypting it transparently at dereference time. By re-purposing the entire 64-bit pointer field for encryption rather than preserving raw address bits, \sys maximizes entropy, eliminates the brute-force weaknesses of truncated authentication codes, and maintains binary compatibility with existing PAC-enabled software. 
We prototype \sys on FPGA using 64-bit RISC-V Rocket and BOOM cores, and evaluate it with microbenchmarks, nbench, and SPEC CPU2017. We compare against both an in-house RISC-V PAC implementation and Apple's PAC on the M1 processor. Across these workloads, \sys provides comprehensive pointer protection with runtime overhead comparable to PAC-based schemes, while incurring negligible area and power overhead. These results show that \sys is a practical design point for deploying strong pointer protection in real processors.

\end{abstract}
\input{1.introduction}

\input{2.background}

\input{3.threat_model}

\input{4.Pointer_encryption}

\input{5.implementation}

\input{6.evaluation}
\input{6.discussion}

\input{7.related}

\input{8.conclusion}
\appendix
\input{9.appendix}


\bibliographystyle{IEEEtranS}
\bibliography{ISCA2025ACMtemplate/sample-base}

\end{document}

%% file: 1.introduction.tex
\section{Introduction}

Modern software systems remain vulnerable to increasingly sophisticated memory corruption exploits. Among the most prevalent and powerful are \emph{control-flow hijacking} and \emph{data-oriented} attacks, which continue to endanger critical infrastructure—from kernels and hypervisors to browsers and database engines. Control-flow hijacking exploits corrupted code pointers such as return addresses, function pointers, or virtual table entries to redirect execution toward attacker-controlled instructions or gadgets. In contrast, data-oriented attacks manipulate data pointers and non-control data to steer legitimate computations toward malicious outcomes without violating the program’s control-flow graph. Together, these attack classes enable arbitrary code execution, privilege escalation, and logic subversion even in hardened environments.

Recent research has explored a wide range of  defenses against the pointer forgery attacks. Broadly, these efforts fall into two categories: \emph{address layout randomization} and \emph{metadata augmentation for integrity}. 
Address layout randomization schemes, such as ASLR~\cite{shacham2004effectiveness}, introduce spatial and temporal unpredictability to memory layouts and thus pointers will be shuffled with memory layout randomization. While effective at increasing attack complexity, these defenses rely on a single random offset for each domain’s address space, providing only limited entropy; once the layout is disclosed, the protection collapses.

Alternatively, metadata can serve either as a cryptographic Message Authentication Code (MAC) checked on pointer use~\cite{mashtizadeh2015ccfi,liljestrand2019pac}, or as a capability/permission that prevents unauthorized pointer modification~\cite{ziad2021zero,woodruff2014cheri,lemay2021cryptographic}. Most defenses~\cite{mashtizadeh2015ccfi,ziad2021zero,woodruff2014cheri} store this metadata in auxiliary structures, providing strong protection but incurring substantial memory overhead or added hardware/ISA management complexity. In contrast, in-place mechanisms~\cite{liljestrand2019pac,lemay2021cryptographic} reuse the 64-bit word containing the pointer itself, eliminating metadata memory overhead.
A prominent realization of in-place pointer integrity protection is the \emph{Pointer Authentication Code} (PAC) mechanism introduced in ARMv8.3-A and later architectures~\cite{qualcomm_pac_v7}. PAC leverages the unused high-order bits of 64-bit virtual addresses (7-16 bits~\cite{lelegard_pac_format}) to store a MAC alongside the pointer, incurring no additional memory overhead~\cite{liljestrand2019pac}.
It uses a lightweight block cipher to generate a short MAC over a pointer, combining a secret key with a contextual modifier.
On dereference, hardware verifies the MAC before use, preventing straightforward pointer corruption and cross-context pointer reuse.
PAC has  been extensively leveraged in the research community to enforce control-flow integrity, type safety, and memory safety~\cite{liljestrand2019pac, ismail2022tightly, farkhani2021ptauth, li2022pacmem, kim2020hardware}, serving as the foundation for numerous hardware-assisted pointer-protection schemes.
In practice, PAC is widely deployed across Arm-based systems to protect return addresses, function pointers, and virtual tables. It secures kernel and user-space control flow on platforms such as Apple \emph{arm64e}, Android, and Windows on Arm \cite{apple_pac_dyld, linux_pac_docs, windows_pac_arm64, clang_pac_docs, apple_PA}. 
These uses make PAC the most prevalent hardware mechanism for enforcing pointer integrity in both commodity and experimental systems.

However, because the authentication code must fit within these unused bits, the effective entropy is small (usually $<24$ bits), making brute-force guessing of valid codes feasible for attackers. 
Other in-place schemes~\cite{lemay2021cryptographic,armmte} are also susceptible to brute-force due to low entropy \cite{na2023penetrating, kim2025tiktag}.
A pointer has 64 bits, and thus, in theory, a protection scheme can raise the brute-force space to $2^{64}$ without additional memory. In the meantime, the information of the pointer value should still be stored in the 64 bits. 
In PAC, pointer values are directly kept in the 64 bits, limiting the brute-force space. On the other hand, 
PAC-protected pointer values are not valid for ordinary use until authentication strips the PAC; arithmetic on the raw pointer value will corrupt the authentication state and cause subsequent checks to fail.

We propose \sys, an architecture that uses a lightweight block cipher in the Electronic Code Book (ECB) mode to {\em fully} encrypt the 64-bit pointers for integrity protection. 
By replacing each 64-bit pointer with its encrypted representation, the architecture can still retrieve the original value when a pointer is dereferenced.  
Fully encrypting the pointer removes the entropy limitation of PAC and prevents attackers from forging valid pointer values through brute-force attacks. Still, \sys's encryption primitive and ISA design provide a similar programming interface to Arm PAC, supporting all protection policies built upon PAC.

In general, encryption alone does not guarantee integrity. Prior work~\cite{madan2005stackoffence,cowan2003pointguard} proposed CTR-mode encryption for pointer protection, but CTR remains vulnerable to targeted bit-flip attacks.
C3~\cite{lemay2021cryptographic} uses partial pointer encryption to detect corrupted pointers, but provides limited security, targeting a 1/16 bypass probability under its threat model. 
To the best of our knowledge, we are the first to evaluate the security and performance of fully encrypting the 64-bit pointer for integrity.

Although encryption and decryption introduce latency overhead, lightweight block ciphers such as PRINCEv2 make full pointer encryption practical with a smaller performance overhead than Arm PAC. 
Compared to MAC-based authentication like PAC, full-pointer encryption delivers much stronger integrity guarantees in both conventional and transient execution, closing the brute-force gap inherent in in-place authentication codes while preserving PAC’s compatibility advantages.
We make the following contributions:

\begin{figure}[t]
    \centering
    \begin{minipage}{0.45\columnwidth}
        \centering
        \begin{lstlisting}[style=CStyle, numbers=none, xleftmargin=0pt]
int VulnFunction(char *p)
{
    char buf[40];
    strcpy(buf, p);
    return 0;
}
        \end{lstlisting}
        \vspace{2mm} 
        \subcaption*{(a) Vulnerable function}
        \vspace{-6mm}
        \label{fig:vul_code}
    \end{minipage}
    \hspace{5mm} 
    \begin{minipage}{0.40\columnwidth} 
        \centering 
        \includegraphics[width=0.7\linewidth]{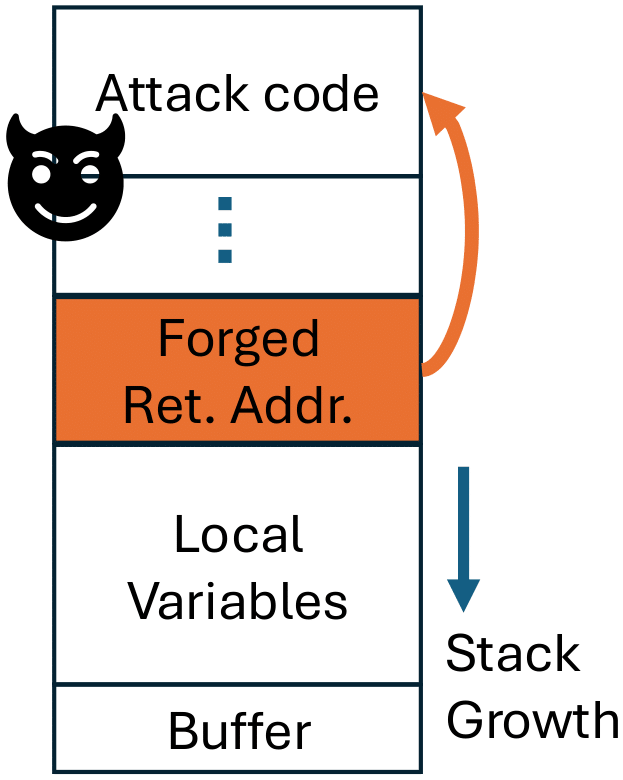}
        \vspace{1mm}
        \subcaption*{(b) Unprotected stack}
        \label{fig:unp_stack}
    \end{minipage}
    \vspace{-1pt}
    \caption{Memory safety exploitation.}
    \label{fig:mem_safty}
\end{figure}

\begin{itemize}
    \item \textbf{Design of \sys.} We propose \sys, a cryptography-based pointer encryption architecture that provides brute-force-resilient protection with lower latency overhead than existing pointer authentication mechanisms. \sys introduces a compatible ISA that leverages existing PAC compiler infrastructure across protection policies. 
    \item \textbf{Co-design system and cipher, Security analysis.} 
    By considering how contextual modifiers are used in practice, we co-design the system and the cipher to avoid using a more expensive tweakable cipher while still providing enough modifier bits.  We formally show that if an attacker were able to forge a valid encrypted pointer, one could construct a proxy adversary capable of launching a chosen-ciphertext attack on the underlying block cipher.

    \item \textbf{Implementation and evaluation.} 
    We implement \sys and a baseline PAC design on RISC-V Rocket and BOOM cores on an FPGA prototype, and additionally examine Apple M1 as a real-world PAC deployment. We evaluate using targeted microbenchmarks, nbench, and SPEC CPU2017, analyzing data-pointer and return-address protection, speculation effects, and compatibility with prior PAC-based compiler passes.
    Results show that \sys achieves performance comparable to or better than PAC while providing substantially stronger security guarantees. Our implementation and compiler are open-sourced at \url{https://github.com/bearhw/LIPPEN}. 

\end{itemize}

%% file: 2.background.tex
\section{Background}

\subsection{Control Flow Hijacking and Data-Oriented Attacks}

C and C++ underpin kernels, hypervisors, browsers, and high-performance libraries since they offer tight control over layout and performance. The same low-level control, however, exposes programs to \emph{memory safety} violations: \emph{spatial} errors (out-of-bounds reads/writes, type confusion) and \emph{temporal} errors (use-after-free, double free, dangling pointers). These defects arise from unchecked pointer arithmetic, manual lifetime management, and implicit casts, and they commonly yield \emph{arbitrary read/write} primitives after exploitation~\cite{unsafeC,ccplusplusbugs}.

Once an attacker can corrupt memory, the next step is often \emph{control-flow hijacking}—diverting the program’s execution to attacker-chosen code or gadgets. Classic stack-based overflows overwrite return addresses or saved frame pointers (\emph{stack smashing})~\cite{smashingstack}, as shown in Figure~\ref{fig:mem_safty}, while heap-based corruptions target function pointers, C++ vtable pointers, longjmp buffers, PLT entries, or indirect branch targets~\cite{vtableattack}. Modern exploits prefer code reuse, chaining short instruction sequences to build \emph{Return/Jump/Call}-oriented programming payloads~\cite{rop, jop, cop}. Memory-safety bugs thus routinely evolve into pointer-forgery primitives that enable both classic and modern exploitation techniques. 

Return-Oriented Programming (ROP)~\cite{rop, carlini:rop} exemplifies how overwriting code pointers or return addresses yields full control of execution without code injection. To counter ROP, Control-Flow Integrity (CFI)~\cite{cfi} was proposed to ensure that execution follows only legitimate paths derived from the program’s control-flow graph, preventing hijacking through corrupted control data such as return addresses or function pointers. However, numerous bypasses have been demonstrated. Counterfeit Object-Oriented Programming (COOP)~\cite{cop} and other CFI-bypass attacks~\cite{cfi-bypass} show how virtual-table and object-pointer corruption can subvert C++ dispatch even when coarse-grained CFI is present, using techniques such as control-flow bending~\cite{carlini:cfb}. More advanced attacks, including Control Jujutsu~\cite{cjujutsu-evans-ccs15}, NEWTON~\cite{van:newton}, and AOCR~\cite{rudd:aocr}, exploit dynamic analysis to bypass even fine-grained CFI protections.

Even when direct control flow is guarded, attackers can compromise program behavior through \emph{data-oriented} and \emph{data-flow} attacks~\cite{dopi, dfa}. Data-Oriented Programming (DOP)~\cite{dopi} shows that corrupting data pointers or non-control data can achieve powerful, semantics-preserving computation without altering branch targets.  
These works show that protecting only control-flow or data values in isolation is insufficient. To substantially raise the bar against modern exploitation, defenders must ensure the integrity of both \emph{code pointers} and \emph{data pointers}, providing comprehensive protection against control-flow hijacking, data-oriented manipulation, and emerging attacks on pointer-authentication schemes.

Recent advances in memory-safe languages such as \emph{Rust} significantly reduce vulnerabilities by enforcing strong ownership and lifetime semantics at compile time. However, Rust cannot eliminate all memory-safety risks: interoperability with legacy C/C++ code, the use of \texttt{unsafe} blocks, and low-level system interfaces can still reintroduce pointer corruption. Moreover, large existing software ecosystems written in C and C++ cannot be easily rewritten in Rust. Consequently, complementary hardware mechanisms that ensure pointer integrity remain essential to securing modern systems end-to-end.

\begin{figure}[t]
    \centering
    \includegraphics[width=0.4\textwidth]{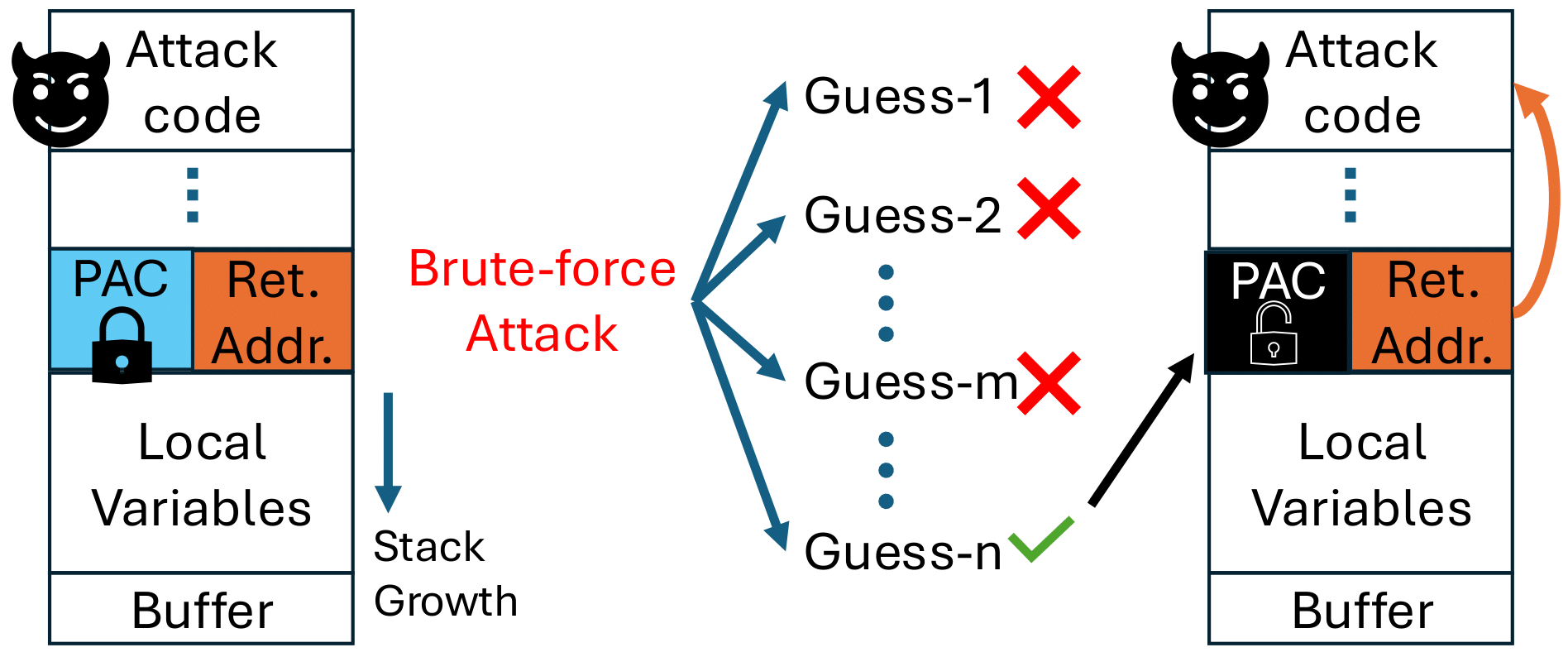}
    \caption{PAC defense and brute-force attack on PAC}
    \label{fig:brute-force}
\end{figure}

\subsection{Pointer Authentication Code} 

Pointer Authentication (PA), introduced in Armv8.3-A, embeds a cryptographic Pointer Authentication Code (PAC) in the unused high-order bits of 64-bit pointers~\cite{arm16}, computed from the \textit{pointer}, a 64-bit \textit{modifier} (context), and a \textit{secret key} via dedicated instructions (e.g., \texttt{PACIA} to sign, \texttt{AUTIA} to verify). On authentication failure, indicating illicit pointer modification, the architecture either corrupts the pointer's top bits to render it unusable (Armv8.3~\cite{apple_PA}) or raises a synchronous exception (Armv8.6~\cite{v8.6}).

By signing return addresses and function pointers, PA thwarts control-flow hijacking attacks such as ROP and JOP, and significantly raises the bar against data-oriented exploits, as attackers must forge valid PAC values to hijack control or data pointers~\cite{arm_rop}.

PAC has served as the foundation for numerous research prototypes that extend pointer authentication to enforce broader security properties such as type safety, temporal safety, and data integrity. Examples include PARTS~\cite{liljestrand2019pac}, which introduces PAC-based control-flow integrity;
PACStack~\cite{PACStack}, which chains authenticated return addresses to strengthen backward-edge protection; 
PTAuth~\cite{farkhani2021ptauth}, which dynamically detects temporal memory corruptions; 
PACMem~\cite{li2022pacmem}, which unifies spatial and temporal protection; 
AOS~\cite{kim2020hardware}, which combines PAC with architectural metadata for object safety; 
PacTight~\cite{ismail2022tightly}, which enforces pointer integrity through strong unique modifiers; and 
RSTI~\cite{ismail2024enforcing}, which systematically generates unique modifiers per pointer based on scope and type.
Beyond research prototypes, PAC has been widely adopted across Arm-based platforms to harden both kernel and user space. Apple’s \emph{arm64e} architecture uses a customized version of PAC to sign return addresses, function pointers, and C++ virtual tables, with compiler and hardware support integrated into iOS, macOS, and their system libraries~\cite{apple_pac_dyld, clang_pac_docs}. Linux and Windows on Arm similarly deploy PAC to protect kernel control-flow structures and sensitive function pointers~\cite{linux_pac_docs, windows_pac_arm64}. The Clang/LLVM toolchain provides first-class support for emitting PAC instructions (\texttt{PACIASP}, \texttt{AUTIASP}, \texttt{BLRAA}) and manages key usage transparently during code generation~\cite{clang_pac_docs}.
Collectively, these systems demonstrate PAC’s versatility as a hardware-assisted substrate for enforcing comprehensive pointer integrity.

\noindent\textbf{Brute-force Attacks on PAC.}
Due to the limited PAC size, an attacker can mount a brute-force attack by exhaustively trying all possible authentication codes (e.g., $2^{16} = 65{,}536$ candidates, reduced further when combined with Memory Tagging Extension (MTE)~\cite{armmte}), as illustrated in Figure~\ref{fig:brute-force}. Such attacks can proceed both conventionally and speculatively.
In a modern processor, instructions, including pointer authentication operations, may execute along mis-speculated paths. If a forged pointer with an incorrect PAC is speculatively used, no architectural fault is raised; instead, the subsequent memory loads fail because authentication has failed. Consequently, the act of speculatively executing pointer authentication and using the resulting pointer can leave detectable microarchitectural effects (for instance, a cache access might occur only if the PAC was correct)~\cite{pacman}.
This creates a potential PAC oracle: a side channel through which an attacker can learn whether a guessed PAC was valid or not, by observing subtle hardware state changes, all without triggering any program-visible error. 
The PACMAN study~\cite{pacman} demonstrated that speculative execution can be exploited to create PAC oracles, enabling systematic brute-forcing of PAC values within a reasonable time (approximately 2.94 minutes on an Apple M1 chip), effectively removing the primary barrier to control-flow hijacking on PAC-protected platforms.

For mitigation, Arm recommends PAC clearing via XPAC immediately after authentication, e.g., AUT; XPAC, which eliminates residual PAC state and prevents speculative dereference of corrupted pointers. 
In later versions, Arm introduces FEAT\_FPACC\_SPEC, making the architectural states not measurably different between failing or passing the authentication~\cite{arm:feat_fpacc}.
However, these fixes mean the forged pointers can be used during transient execution, not protecting the pointers in Spectre attacks~\cite{na2023penetrating}. Attackers can still launch attacks like speculative ROP~\cite{kiriansky2018speculative}.

%% file: 3.threat_model.tex
\section{Threat model}

\begin{table*}[t]
\centering
\begin{threeparttable}
\caption{\diff{Comparison of representative pointer integrity defenses.}}
\label{tab:pointer_comparison}
\renewcommand{\arraystretch}{1.05}
\setlength{\tabcolsep}{6pt}

\begin{tabular}{|
  >{\raggedright\arraybackslash}p{0.14\textwidth} |
  >{\raggedright\arraybackslash}p{0.10\textwidth} |
  >{\centering\arraybackslash}p{0.15\textwidth} |
  >{\centering\arraybackslash}p{0.13\textwidth} |
  >{\centering\arraybackslash}p{0.18\textwidth} |
  >{\centering\arraybackslash}p{0.13\textwidth} |
  }
\hline
\textbf{Category} &
\textbf{Defense} &
\textbf{Brute-force Resilience} &
\textbf{Memory Footprint Overhead} &
\textbf{Required Modification} &
\textbf{Context(Permission) / pointer} \\ \hline

\multirow{2}{*}{Address Randomization} &
ASLR~\cite{shacham2004effectiveness} &
19--28 bits\tnote{a} &
None &
OS &
No \\ \cline{2-6}
& Morpheus~\cite{gallagher2019morpheus} &
60 bits\tnote{a} &
2 bits/word &
Crypto Engine + Churn-Unit + OS &
No \\ \hline

\multirow{5}{*}{External Metadata} &
CHERI~\cite{woodruff2014cheri} &
\multirow{4}{*}{\parbox[c]{0.15\textwidth}{\centering Metadata not accessible in user space.\\ Cannot be brute forced}} &
256 bits/pointer &
\multirow{4}{*}{\parbox[c]{0.18\textwidth}{\centering
\begin{tabular}{@{}l@{}}
Instructions + compiler\\ + memory hierarchy
\end{tabular}}} &
\multirow{4}{*}{\parbox[c]{0.13\textwidth}{\centering
\begin{tabular}{@{}l@{}}
Included in tags
\end{tabular}}} \\ \cline{2-2} \cline{4-4}
& ZeRØ~\cite{ziad2021zero} &
 &
2 bits/word &
 &  \\ \cline{2-2} \cline{4-4}
& PUMP~\cite{dhawan2015architectural} &
 &
word-size bits/word &
 &  \\ \cline{2-2} \cline{4-4}
& Star\cite{gollapudi2023control} &
 &
(2,6) bits/word for (Data, Instructions) &
 &  \\ \cline{2-6}
& CCFI~\cite{mashtizadeh2015ccfi} &
128 bits &
128 bits/pointer &
Compiler\tnote{b} &
80 bits \\ \hline

\multirow{4}{*}{In-pointer Metadata} &
FRP~\cite{phaye2025fully} &
52 bits &
16 bytes/object &
OS + memory hierarchy &
No \\ \cline{2-6}

& PAC~\cite{liljestrand2019pac} &
7--16 bits &
\multirow{3}{*}{\centering None} &
\multirow{3}{*}{\parbox[c]{0.18\textwidth}{\centering
Instructions + Compiler\tnote{c} + \diff{Hardware }Crypto Engine}} &
64 bits \\ \cline{2-3} \cline{6-6}

& C3~\cite{lemay2021cryptographic} &
24 bits &
 &  & Yes \\ \cline{2-3} \cline{6-6}

& \sys (ours) &
64 bits &
 &  & Adjustable\tnote{d} \\ \hline

\end{tabular}

\begin{tablenotes}
\footnotesize
\item[a] This entropy accounts for the whole system's security and is different from per pointer entropy.
\item[b] They use Intel AES-NI engine for encryption.
\item[c] Our work (\sys) uses the same Instructions and Compiler support as Arm PAC. 
\item[d] up to 192 bits. maximum level of security is achieved if we keep context size as large as the unused bits. 
\end{tablenotes}

\end{threeparttable}
\end{table*}

We follow a threat model similar to typical memory-corruption attacks~\cite{mashtizadeh2015ccfi}. The attacker aims to exploit memory-safety vulnerabilities, such as stack buffer overflows and use-after-free bugs, to corrupt pointers and thereby subvert control flow or mount data-oriented attacks. We assume a powerful attacker who, after exploiting such a vulnerability, can perform arbitrary reads of process memory and overwrite any writable memory location, including code pointers, data pointers, and user-space protection metadata such as modifiers. This capability enables both control-flow hijacking and data-oriented manipulation.

We also include pointer forgery during transient execution attacks~\cite{kiriansky2018speculative,na2023penetrating,goktas2020speculative} in our threat model. If pointer integrity is violated during speculation (e.g., Speculative ROP~\cite{kiriansky2018speculative}, or forging a pointer in Spectre attacks), we consider it an attack. However, we focus on pointer integrity protection; other information leakages due to side channels not using pointer forgery are out of scope.

We trust the underlying hardware and operating system kernel, which are responsible for securely generating, managing, and storing the process-wide secret key \(K\), which remains constant during execution and is inaccessible to user-level code.


%% file: 4.Pointer_encryption.tex
\section{\sys Design}
\subsection{Design Goals}

\sys is designed to ensure robust pointer integrity while maintaining practicality and efficiency in both conventional and speculative execution. The key design goals are:

\begin{enumerate}[label=\textbf{G\arabic*:}, leftmargin=1.2cm]
    \item \textbf{Comprehensive Pointer Integrity Coverage.} Protect all pointer types, including both data pointers and code pointers (e.g., return addresses).

    \item \textbf{Zero Metadata Overhead.} Eliminate auxiliary data structures, shadow memory, or tag tables to avoid memory overhead and access latency.

    \item \textbf{High Security Strength.} Provide strong protection against pointer corruption, reuse, and brute-force attacks by maximizing cryptographic entropy within the pointer representation.

    \item \textbf{Low Performance Overhead.} Use lightweight, hardware-assisted encryption to achieve near-PAC runtime overhead, enabling deployment in performance-sensitive systems. 

    \item \textbf{Flexible Hardware Support for Various Security Policies.} Support multiple protection policies and context-binding schemes under user or compiler control, e.g., modifier in PAC, enabling flexible trade-offs between performance, security, and compatibility.

    \item \textbf{Reuse of the Existing Toolchain for Easy Deployment of Defenses} Preserve existing compiler instrumentation, ABI conventions, and runtime interfaces so that PAC-enabled software can run on \sys without modification, ensuring drop-in integration into existing software and toolchains.
\end{enumerate}

\subsection{Design Space Discussion}

A wide range of mechanisms have been developed to defend against memory corruption and control-flow attacks. Table~\ref{tab:pointer_comparison} compares representative designs across four axes: brute-force space, memory footprint overhead, deployment requirements, and context granularity within each domain. We analyze their trade-offs and how they align with our design goals.
Broadly, these mechanisms fall into two primary categories: \emph{address layout randomization} and \emph{metadata augmentation}.

\paragraph{Address Layout Randomization}
Address-randomization defenses~\cite{aslr,gallagher2019morpheus} increase attacker uncertainty by randomizing the placement or representation of code and data objects.  Conventional ASLR is practical and widely deployed because it requires no changes to program binaries or pointer formats, but it provides only probabilistic protection: once layout entropy is disclosed or guessed, leaked pointers can be reused to construct code-reuse or control-flow hijacking attacks. Prior work~\cite{na2023penetrating} further shows that such secret offsets can be inferred through speculative probing attacks~\cite{goktas2020speculative}. Morpheus~\cite{gallagher2019morpheus} strengthens this class of defenses with hardware-supported moving-target defenses and runtime churn. It applies a random displacement to code and data pointers by adding a secret offset to their values, and periodically changes these offsets during execution (e.g., every 50\,ms in the evaluated configuration) so leaked or brute-forced information becomes stale. This stronger protection, however, requires substantial architectural support, including 2-bit runtime domain tags per 64-bit word, additional tag storage/cache structures, and specialized hardware for churning, pointer translation, and attack detection.
Thus, while address randomization is an effective baseline defense, conventional ASLR does not provide pointer integrity (\textbf{G1}), and stronger variants such as Morpheus achieve higher security only with nontrivial hardware and metadata complexity, falling short of our goals for comprehensive coverage (\textbf{G1}) and zero metadata overhead~(\textbf{G2}).

\paragraph{Metadata Augmentation}
A second family of techniques strengthens pointer integrity by associating pointers with auxiliary integrity or capability metadata~\cite{mashtizadeh2015ccfi, ziad2021zero, woodruff2014cheri, liljestrand2019pac, gollapudi2023control, sasaki2019practical, lemay2021cryptographic, cowan2003pointguard, phaye2025fully}.
Depending on where metadata is maintained, these defenses can be grouped into two categories:

\noindent \textbf{Protection with external auxiliary metadata structures.} Systems such as CCFI~\cite{mashtizadeh2015ccfi}, ZeRØ~\cite{ziad2021zero},
Star~\cite{gollapudi2023control}, and capability-based architectures like CHERI~\cite{woodruff2014cheri} associate pointers with auxiliary metadata structures or tagged memory that encode authentication codes, bounds, or permissions. Such schemes provide robust protection, comprehensive coverage (\textbf{G1}), and strong security guarantees (\textbf{G3}) by maintaining precise, per-pointer metadata. However, they incur nontrivial memory, lookup, and synchronization overheads. Their reliance on external structures complicates hardware design and violates the \textbf{zero-metadata} principle (\textbf{G2}), making them less practical for lightweight or commodity environments.

\noindent \textbf{In-Place Protection.}  Several works, such as Arm Pointer Authentication (PAC)~\cite{arm16} and C3~\cite{lemay2021cryptographic}, embed integrity metadata directly into the pointer representation by repurposing unused high-order address bits. Such approaches have zero metadata overhead and are adopted by industry~\cite{arm16}. However, subsequent analyses have demonstrated that such schemes are vulnerable to brute-force or collision-based attacks: PAC is compromised by PACMAN~\cite{pacman}, and C3 is bypassed by Na \textit{et al.}~\cite{na2023penetrating}. The underlying reason is that the number of authentication bits embedded within the pointer is limited; PAC uses only 11–15 bits, and C3 employs a 24-bit cipher, significantly constraining the available integrity space.
\diff{ Similarly, FRP~\cite{phaye2025fully} encodes 52 bits of the pointer, including the unused high order bits; however, its encoding/decoding uses table (map) lookup instead of cryptography, introducing extra indirection and performance overhead.  Without new instructions, it relies on malloc to manage the map and supports only heap objects.
}

\textbf{\em Pointer Authentication.}
Current in-place schemes, such as Arm PAC~\cite{liljestrand2019pac}, retain most pointer bits for the raw address and dedicate only a small fraction of high-order unused bits to the authentication code. This design choice stems from two considerations by Arm~\cite{qualcomm_pac_v7}: (i) preserving the full address value allows pointers to participate in branch prediction without authentication, and (ii) maintaining meaningful address values simplifies software debugging and crash analysis. These choices, however, dramatically reduce the number of bits available for authentication—typically to fewer than two dozen bits (16 bits on the Apple M1~\cite{pacman}), depending on the virtual address space layout. As demonstrated by the PACMAN attack~\cite{pacman}, such limited entropy enables practical brute-force and oracle-based attacks that can recover valid authentication codes within a few minutes on commodity hardware.


One benefit of keeping the raw pointer value is that it leaves room for micro-architectural optimizations, such as using the pointer speculatively before the authentication completes, hiding the authentication latency. 
e.g., Arm also provides the fused instruction for authentication and load (\texttt{LDRAA} and \texttt{LDRAB}) and authentication and return (\texttt{RETAA} and \texttt{RETAB}).
However, in practice, as shown in the PACMAN~\cite{pacman} attack, the authentication result impacts the execution result of the follow-up instruction consuming the pointer, 
indicating limited overlap between the load or return operation and authentication. Otherwise, if a data pointer is used speculatively before authentication, the TLB must be updated regardless of the authentication result, since it lies on the critical path. This prevents the PACMAN attack from succeeding on data pointers. Additionally, PARTS~\cite{liljestrand2019pac} uses additional instructions to emulate the 4-cycle authentication delay for end-to-end performance evaluation, showing reasonable performance without speculation. We corroborate this with Apple M1 measurements, where PAC-protected data pointer accesses in a pointer-chasing scenario incur overhead comparable to encryption where no raw pointer bit exists.
 
Arm also provides the XPAC instruction set (e.g., \texttt{XPACI}, \texttt{XPACD}, and \texttt{XPACLRI}) to strip the PAC from a pointer and recover its original address~\cite{armxpac}. However, XPAC operations are typically invoked only in specialized contexts where security is not a concern, such as specific pointer arithmetic, low-level runtime code, or debugging, and are rarely executed in normal application paths.

\textbf{\em This paper: Full Pointer Encryption.}
These findings suggest that preserving full raw address bits in authenticated pointers may not be necessary in practice while severely constraining available entropy. 
\sys leverages this insight to address the core limitations of prior designs. By repurposing all 64 bits of the pointer for cryptographic protection, \sys achieves comprehensive coverage across all pointer types (\textbf{G1}), eliminates external metadata to meet the zero-overhead requirement (\textbf{G2}), and maximizes entropy to provide strong cryptographic protection against brute-force and reuse attacks (\textbf{G3}). This full-pointer encryption approach removes the entropy bottleneck inherent to truncated MACs and transparently restores a valid address upon dereference, providing robust integrity and confidentiality without additional memory or hardware state.

However, full pointer encryption still faces challenges in finding a suitable cipher, conducting a security evaluation, and providing debugging support when needed.

\begin{table*}[t]
  \centering
  \caption{Lightweight cipher candidates for \sys. \diff{Area values are representative gate equivalents (GE) as reported in the original cipher publications or widely cited hardware implementation studies for compact round-based implementations. 
  C}ycle counts correspond to typical round-based hardware implementations. }
  \setlength{\tabcolsep}{2pt} 
  \label{tab:lit-cipher-comparison}
  \begin{threeparttable}
  \begin{tabular}{lcccccccl@{}}
    \toprule
    \textbf{Cipher} &
    \makecell{\textbf{Block Size}\\\textbf{(bits)}} &
    \makecell{\textbf{key}\\\textbf{(bits)}} &
    \makecell{\textbf{tweak}\\\textbf{(bits)}} &
    \makecell{\textbf{Rounds}\\\textbf{}} &
    \makecell{\textbf{Area}\\\textbf{(GE)}} &
    \makecell{\textbf{Latency}\\\textbf{(cycles or ps)}} &
    \textbf{Notes relevant for \sys} \\
    \midrule
    \multicolumn{8}{c}{\emph{Early lightweight / \diff{area-optimized} ciphers (not latency-optimized, area = $\sim$ listed area cost X no. of rounds)}} \\
    \midrule
    KATAN/KTANTAN~\cite{canniere2009katan}     & 64  & 80-bit     & — & 254 & \hspace{-22pt}$\sim$3200/$\sim$3000\hspace{-22pt}      & $\sim$254 cyc & Bit-serial, ultra-low-area, long latency. \\
    LED-64~\cite{leander2011led}         & 64  & 64/128-bit & — & 32/48 & $\sim$1200     & 32 cyc       & Compact SPN; unrolled path too deep. \\
    Piccolo-80/128~\cite{biryukov2011piccolo}     & 64  & 80/128-bit     & — & 25/31  & 683/758             & 432/528 cyc      & Very small area; serialized 4-bit datapath. \\
    TWINE-80~\cite{suzaki2012twine}       & 64  & 80-bit     & — & 36  & 1503            & 36 cyc       & Round count too large for 1-cycle use. \\
    KLEIN-64/80/96~\cite{gong2012klein}       & 64  & 64/80/96-bit     & — & 12/16/20  & 1981/2097/2213            & 105/107/109 cyc     & Compact SPN; Not optimized for unrolling. \\
    PRESENT-80~\cite{bogdanov2007present}     & 64  & 80-bit     & — & 31  & 1570            & 32  cyc      & ISO lightweight cipher; unrolled area large. \\
    SIMON64/128~\cite{beaulieu2015simon}    & 64  & 128-bit    & — & 44  & $\sim$1200--1400 & 44  cyc      & Hardware-oriented Feistel; too many rounds. \\
    ASCON (perm.)~\cite{dobraunig2021ascon}  & \hspace{-20pt}320 (rate 64/128)\hspace{-22pt} & 128-bit & — & 6/8 perms & $27280$ & 6 cyc & AEAD permutation; not a 64-bit block cipher. \\
    \midrule
    \multicolumn{8}{c}{\emph{Low-latency ciphers (suitable candidates for pointer authentication, e: encryption, d: decryption)}} \\
    \midrule
    K-Cipher~\cite{kounavis2020kcipher}\tnote{a}       & var.& param.     & — & 10--14       & 42552(e) & 767ps/2--3 cyc & High-throughput low-latency via pipelining. \\
    BipBip~\cite{belkheyar2022bipbip}\tnote{a}         & 24  & 128-bit    & 128-bit & 7 (stages)   & 5741(d) & 622(d) ps/3 cyc & Depth-3 pipeline; block size too small. \\    
    QARMA-64~\cite{avanzi2017qarma} \tnote{b}      & 64  & 128-bit & 64-bit & 11       & 22131(e/d) & 553 ps & Critical path larger than PRINCE-family. \\
    PRINCE~\cite{borghoff2012prince}\tnote{b}         & 64  & 128-bit            & — & 12 (5+mid+5) & 13468(e/d) & 401 ps  & Designed for single-cycle unrolled latency. \\
    PRINCEv2~\cite{princev2}\tnote{b}       & 64  & 128-bit            & — & 12            & 14181(e/d) & 404 ps & Best latency/area trade-off for \sys. \\
    \bottomrule
  \end{tabular}
  \begin{tablenotes}
    \footnotesize
    \item[a] The numbers for K-Cipher and BipBip are taken from the original publications \cite{kounavis2020kcipher, belkheyar2022bipbip}. Note that we normalized the area result for K-Cipher as GE based on Intel 10\,nm library characteristics provided in \cite{8268472}, it is originally reported as 1875$\mu\text{m}^2$.We report decryption-only numbers for BipBip as highlighted also in the original work.
    \item[b] The results for QARMA, PRINCE, and PRINCEv2 are taken from the original PRINCEv2 work \cite{princev2}, as all ciphers were implemented in NanGate 15\,nm technology setting, which provided us with a fairer comparison. Note that, according to our calculations based on NanGate 15\,nm Open Cell Library Databook, 13468 GE translates to 3391$\mu\text{m}^2$ for PRINCE and 14181 GE translates to 3570$\mu\text{m}^2$ for PRINCEv2. We report only the results for e/d shared datapath architectures in the table, encryption-only results are slightly smaller/faster than these.
  \end{tablenotes}
  \end{threeparttable}
\end{table*}

\begin{table}[t]
  \centering
  \caption{Cipher area and timing (post-implementation).}
  \label{tab:cipher-area-timing}
  \setlength{\tabcolsep}{6pt}
  \begin{tabular}{lcccc}
    \toprule
    \textbf{Cipher} & \textbf{LUTs} & \textbf{FFs} &
    \makecell{\textbf{Latency}\\\textbf{(cycles)}} &
    \makecell{\textbf{Fmax}\\\textbf{(MHz)}} \\
    \midrule
    QARMA-vhd~\cite{memsec_hdl_crypto}    & 1670 & 65 & 2 & 67 \\
    PRINCE-vhd~\cite{memsec_hdl_crypto} & 1233 & 65 & 2 & 84 \\
    PRINCE-vhd+mod & 1250 & 65 & 2 & 72 \\
    QARMA-unrolled-verilog    & 1794 & 0 & 1 & 40 \\
    PRINCE-unrolled-verilog~\cite{borghoff2012prince} & 1378 & 0 & 1 & 41 \\
    PRINCEv2-unrolled-verilog& 1378 & 0 & 1 & 44 \\
    PRINCEv2-unrolled-verilog+mod & 1522 & 0 & 1 & 42 \\
    \bottomrule
  \end{tabular}
\end{table}

\subsection{Choice of Cipher}

A central requirement in \sys's design is to achieve strong cryptographic protection with minimal performance impact (\textbf{G4}). 
As pointer unsealing and sealing occur frequently along critical execution paths, the encryption primitive must be ``light'' enough in terms of operation latency, area, and power. 
Since \sys aims to protect the pointers in 64-bit machines without adding additional memory footprint (\textbf{G2}), we consider ciphers that operate on a message size of 64-bit or less. 
Over nearly two decades, the cryptographic community has proposed numerous 64-bit lightweight block ciphers targeting low-area or low-power implementations, and
Table~\ref{tab:lit-cipher-comparison} summarizes the intended design goals, security levels (key sizes), block sizes, and expected latency and area characteristics of the widely-used ciphers among these proposals.

Early proposals such as KATAN and KTANTAN~\cite{canniere2009katan}, LED~\cite{leander2011led}, Piccolo~\cite{biryukov2011piccolo}, TWINE~\cite{suzaki2012twine}, and KLEIN~\cite{gong2012klein} demonstrate this direction of design clearly. Prominent lightweight block ciphers such as the ISO-standard \textsc{PRESENT}~\cite{bogdanov2007present}, the NIST lightweight authenticated encryption standard \textsc{ASCON}~\cite{dobraunig2021ascon}, and the NSA's SIMON family~\cite{beaulieu2015simon} follow similar principles. However, while these designs offer small round-based hardware footprints, their large number of rounds renders them impractical for unrolled implementations, where the full datapath is required in a single cycle. Their corresponding critical paths are too long for tightly integrated pointer authentication on the processor's load–use path (see Table~\ref{tab:lit-cipher-comparison}).

For \sys, which performs a pointer integrity check (i.e., decryption) on all protected pointers before use, the decryption latency directly influences the critical path of the program. 
This requirement moves the design space away from classical lightweight designs and toward \emph{low-latency} block ciphers explicitly engineered for unrolled hardware implementations. Examples include  K-Cipher \cite{kounavis2020kcipher}, BipBip \cite{belkheyar2022bipbip}, QARMA \cite{avanzi2017qarma}, \textsc{PRINCE} \cite{borghoff2012prince}, and \textsc{PRINCEv2} \cite{princev2}, all of which target ultra-short critical paths. 
Designs such as K-Cipher and BipBip offer deeply pipelined low-latency modes (e.g., depth-3 pipeline at around 4-4.5\,GHz in 10\,nm technology for both ciphers \cite{kounavis2020kcipher, belkheyar2022bipbip}), which are useful for high-throughput applications. 
However, the security level of BipBip is aligned for encrypting 24-bit blocks in every encryption, which is not suitable for \sys's encryption of 64-bit blocks requirement for pointer protection \cite{belkheyar2022bipbip}. Furthermore, K-Cipher's parameterizable structure and few rounds do not provide the same level of publicly scrutinized provable security margins as PRINCE-family ciphers \cite{mahzoun2022kcipherattack}.
In terms of area, cost, and security level trade-off, PRINCE-family ciphers still perform better.
This is best reflected in the results table presented in the original PRINCEv2 cipher manuscript (Table 6, \cite{princev2}), where the authors provide the latency and area results for both PRINCE and PRINCEv2 in NanGate 15\,nm Open Cell Library (see Table~\ref{tab:lit-cipher-comparison} for details): the results presented here are comparable to K-Cipher; especially when the area cost of unrolling is taken into account. In the best latency setting (401 ps) for PRINCE, the area consumption in $\mu\text{m}^2$ is more than K-Cipher (when the numbers in Table 6 in \cite{princev2} are translated back to $\mu\text{m}^2$ according to the NanGate 15\,nm Open Cell Library Databook); however, the area cost for PRINCE with a latency of 600 ps is comparable to K-Cipher with the same calculation.

In contrast, \textsc{QARMA}, \textsc{PRINCE} and \textsc{PRINCEv2} are explicitly optimized for one-cycle unrolled datapaths under reasonable area budgets. QARMA~\cite{avanzi2023qarmav2} offers a 64-bit tweak, which can encode context for pointer protection; however, it demands significantly larger area and suffers from longer logic depth in unrolled mode than PRINCE. The PRINCE-family provides strong built-in support for ``encryption = decryption'' with minimal overhead and constant-time structure design.  PRINCEv2 offers a more suitable balance of latency, area, and security margins for pointer integrity. However, without a tweak structure like QARMA, a scheme is needed to incorporate the context for pointer protection.

\diff{As the values in Table~\ref{tab:lit-cipher-comparison} are collected from each individual paper and are not normalized across works, it only provides an approximate comparison of design complexity. }
For the ciphers considered in \sys, we implemented representative designs on an AMD/Xilinx VCU118 FPGA platform with a Virtex UltraScale+ device and report post-implementation area and timing in Table~\ref{tab:cipher-area-timing}. We use publicly available HDL implementations as baselines: the QARMA Verilog implementation is adapted from the corresponding VHDL design, while the PRINCE and PRINCEv2 implementations are adapted from~\cite{borghoff2012prince}.
The results show that the unrolled PRINCE-family implementations are the best fit for \sys's latency and area requirements. Relative to unrolled QARMA, PRINCEv2 requires fewer LUTs and achieves a slightly higher maximum frequency, while maintaining single-cycle latency. The \texttt{+mod} variants add the XOR logic needed to inject a design-specific tweak into the PRINCE datapath. Although this modification increases area relative to the unmodified PRINCE-family designs, PRINCEv2+\texttt{mod} remains substantially smaller and faster than unrolled QARMA. We therefore use \textsc{PRINCEv2+mod} as the pointer-sealing cipher in \sys, as it provides single-cycle sealing with low area overhead and directly supports the tweak integration required by our design.

\subsection{ISA and Programming Interface Design}
\sys aims to introduce strong pointer protection, while trying to reuse existing software infrastructures rather than redesigning them. Thus, \sys tries to realize full-pointer encryption in a way that preserves PAC's practical advantages (compact representation, compiler and ABI compatibility, context binding, and in-line hardware operation) while fundamentally elevating its security guarantees.
In PAC, the use of context (modifier) is essential to security: it cryptographically ties each pointer to its creation environment, such as a stack frame, privilege level, or protection domain, so that even if a pointer is leaked or copied, it cannot be validly reused in another context (cross-domain pointer reuse attack). 
\sys keeps a similar protection mechanism and is compatible with existing compiler instrumentation, while establishing the foundation for the subsequent goals of performance efficiency (\textbf{G4}), configurability (\textbf{G5}), and seamless deployability (\textbf{G6}).

In the pointer encryption design, \sys treats every pointer as an encrypted capability. The encrypted pointer will be similar PAC-protected pointers in Arm.  Each pointer is encrypted with a  context when generated or stored and decrypted only at the point of dereference, ensuring integrity and confidentiality throughout its lifetime.

To achieve \textbf{G6} (\emph{Seamless Deployability}), \sys integrates full-pointer encryption into the instruction set architecture with minimal disruption to existing software and toolchains. The design extends the ISA with a small set of new instructions that provide efficient hardware interfaces for pointer encryption and decryption. The system manager or operating system configures these settings using the \texttt{SET\_KEY} and \texttt{SET\_M\_SIZE} instructions. Here, \texttt{SET\_KEY} assigns the 128-bit encryption key to the current security domain, while \texttt{SET\_M\_SIZE} specifies the configurations of the modifier components. 

Applications can use the \texttt{PTR\_SEAL} and \texttt{PTR\_UNSEAL} instructions, which provide an efficient interface for encrypting and decrypting pointers. Their semantics closely mirror PAC's \texttt{PAC*} and \texttt{AUT*} instructions, allowing direct reuse of existing compiler instrumentation, APIs, and ABI conventions without modification. Table~\ref{tab:isa} summarizes their functionality.
Also, for debugging purposes, the protection can be turned off.

\begin{table}[t]
\centering
\caption{Pointer Encryption ISA Extensions}
\label{tab:isa}
\setlength{\tabcolsep}{4pt} 
\begin{tabular}{|l|p{4.8cm}|}
\hline
\textbf{Instruction} & \textbf{Description} \\ \hline
\texttt{SET\_KEY(K1, K2)} & Assigns the 128-bit encryption key by concatenating $K_1$ and $K_2$. \\ \hline
\texttt{SET\_M\_SIZE(conf)} & Sets the configurations for modifier use. \\ \hline
\texttt{PTR\_SEAL(ptr, mod)} & Encrypts a 64-bit pointer using the system-managed secret key and an optional modifier. Invoked when a pointer is created or stored. \\ \hline
\texttt{PTR\_UNSEAL(ptr, mod)} & Decrypts and validates an encrypted pointer before dereference. \\ \hline
\end{tabular}
\end{table}

\begin{figure}[t]
  \centering
  \subfloat[Encryption\label{fig:enc}]{
    \includegraphics[width=0.45\linewidth]{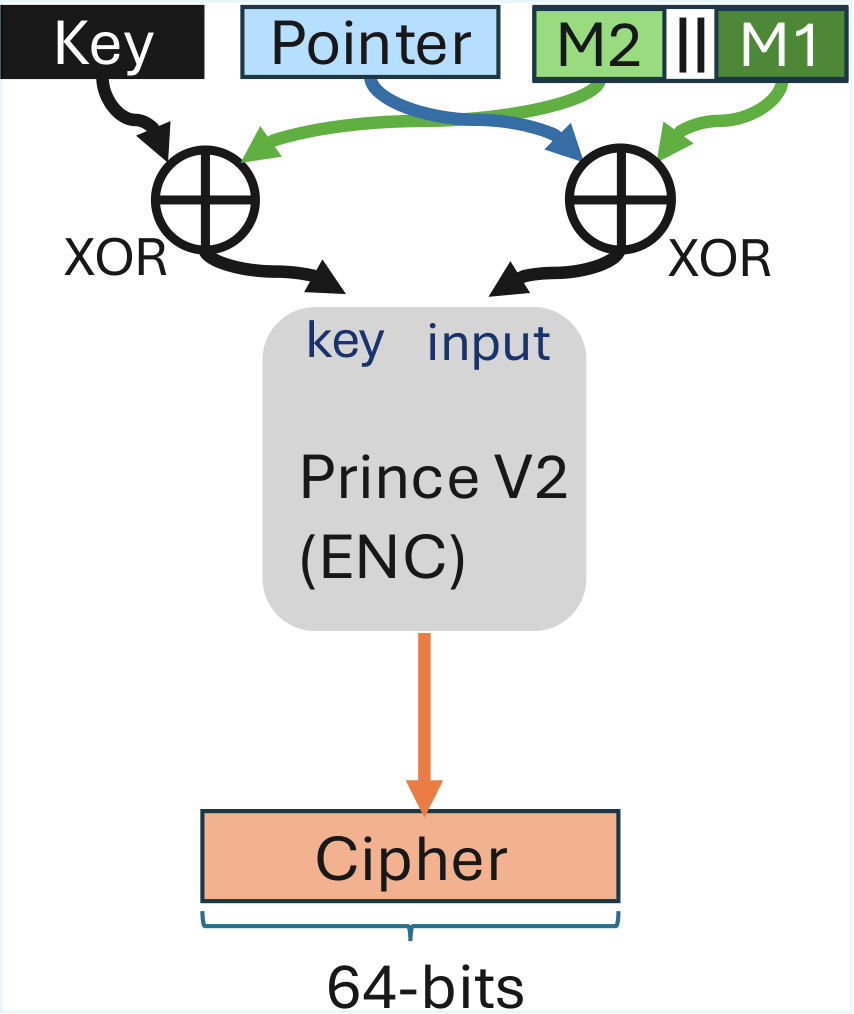}
  }\hfill
  \subfloat[Decryption\label{fig:dec}]{
    \includegraphics[width=0.45\linewidth]{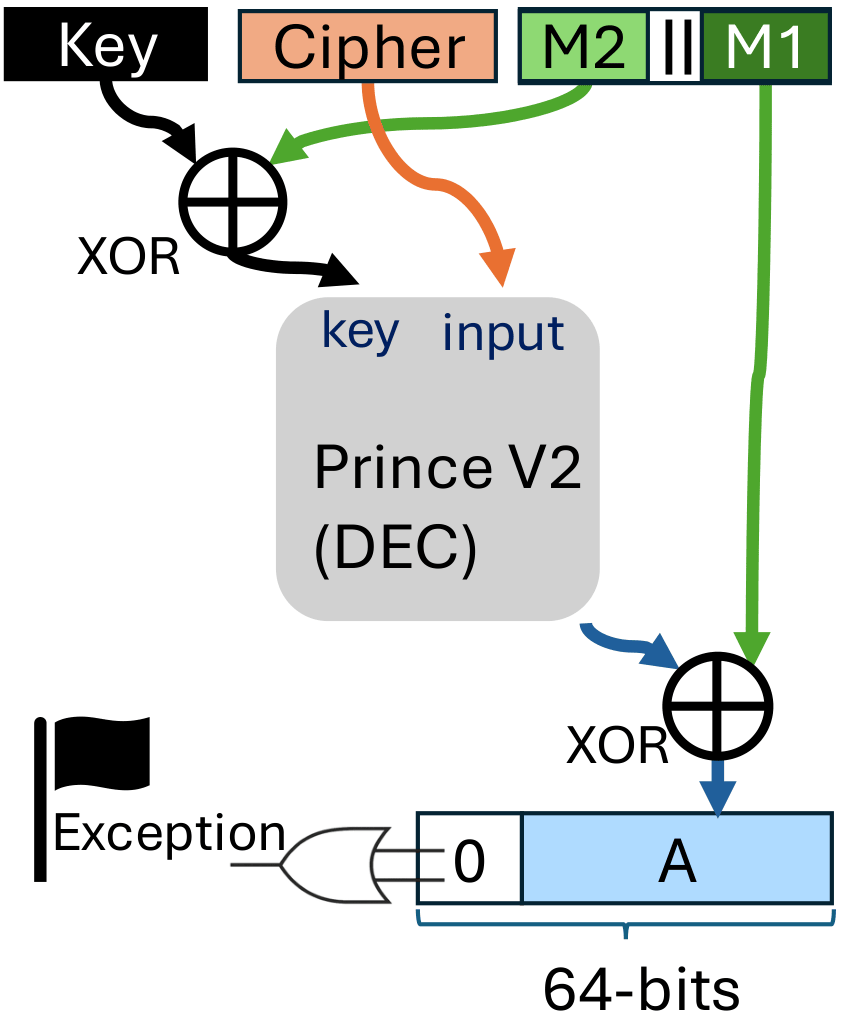}
  }
  \caption{Pointer encryption and decryption design with PRINCEv2. Key has 128 bits, input/output 64 bits, and M1 and M2 sizes are user defined. \diff{After decryption completes, we expect the unused bits of the resulting plaintext to be all zeros, otherwise an exception flag will be raised.}}
  \label{fig:encdec}
\end{figure}

\subsection{Modifier Design}
To address \textbf{G5} (\emph{Configurable Protection Modes}) and \textbf{G6} (\emph{Reusing Existing PAC Toolchain}),
\sys introduces a flexible modifier design that enables fine-grained control over how pointers are cryptographically bound to their execution context. A key principle in pointer encryption is that each \emph{encrypted pointer} must depend on three components:  
(i) a secret key managed by the system and isolated from software control,  
(ii) the pointer value itself, and  
(iii) a context-dependent modifier that ties the encrypted pointer to specific execution conditions, preventing its reuse in unauthorized domains.

\subsubsection{Design Challenges}
As shown in Table~\ref{tab:lit-cipher-comparison} and \ref{tab:cipher-area-timing}, PRINCEv2 has the best latency among the 64-bit block ciphers. Ciphers with tweaks incur higher latency and area to process the tweaks. 
Can we co-design system and cipher to reduce the need for tweaks and lower the protection latency?

\subsubsection{Design Approach.}
PRINCE encryption has plaintext message and key as the input. So the design options are to mix the modifier into the plaintext or key.
Specifically, we define the encryption and decryption as:
\[
\mathit{cipher} = \mathsf{seal}(k, \mathit{ptr},m)=\mathsf{Enc}_{k \oplus m_2}(\mathit{plain} \oplus m_1)
\]
\[
\mathit{plain ptr} = \mathsf{unseal}(k, \mathit{cipher},m)=\mathsf{Dec}_{k \oplus m_2}(\mathit{cipher}) \oplus m_1
\]
Here, $m=m_1 || m_2$ represents modifier components derived from the execution context (e.g., privilege level, address-space identifier, or control-flow epoch), as illustrated in Figure~\ref{fig:encdec}.
The lengths of $m_1$ and $m_2$ can be configured by the system manager using the \texttt{SET\_M\_SIZE($m_1$, $m_2$)} instruction, as described in Table~\ref{tab:isa}. 
We modified the PRINCE and PRINCEv2 source codes to incorporate the XOR logic for $m_1$ and $m_2$ and compared their maximum frequency and area with QARMA, as shown in \texttt{+mod} rows in  Table~\ref{tab:cipher-area-timing}, PRINCEv2 continues to deliver the best performance.

\subsubsection{Modifier Design Security Analysis}
\label{sec:mod_design}
~

\paragraph{Attacker Assumptions}
Based on our threat model, the attacker may exploit software vulnerabilities to overwrite {\em protected pointer} and {\em modifier} values at runtime. However, the attacker does \emph{not} have access to or control the key \(K\) for the victim process.  We also assume the attacker may have access to a set of observed tuples
$\mathcal{T} = \{(m_i, p_i, c_i)\mid c_i = \mathsf{Enc}_{K}(p_i)$ potentially by observing the victim program. The goal of the attacker is to forge an encrypted pointer that dereferences to a target memory location $p_a$, thereby violating pointer integrity.
That means a successful attacker 
can create a valid tuple
$(m_a, p_a, c_a)$
such that \(c_a\) decrypts correctly under modifier \(m_a\) to an attacker-chosen pointer \(p_a\), i.e.,
$p_a = \mathsf{Dec}_{K \oplus m_{2,a}}(c_a) \oplus m_{1,a}.$

\paragraph{Potential Target Bit Flip Attacks with $m_1$}
~

If $m_1$ overlaps with the pointer $ptr$, and $X$ denote the bitmask for the overlapping bits.
Then the attacker can construct $(m_{1,a} = m_{1,v} \oplus X, p_a = p_v\oplus X, c_a = c_v)$ where  $(m_v, p_v, c_v)$ is a valid tuple. As a result, the attacker can overwrite the modifier to flip certain bits in the pointer if $X$ is not zero.

To avoid this ambiguity, \(m_1\) should be placed only in pointer bits that do not affect address generation. On most 64-bit architectures, only the lower \(A\) bits are used as virtual-address bits, e.g., \(A=48\) for a 48-bit virtual address space; the remaining \(64-A\) high-order bits are unused or sign-extension bits. In addition, if pointers are word-aligned, the two least-significant bits do not affect the addressed word. These bits can therefore be safely used for \(m_1\), yielding
\[
    |m_1|_{\max} = 64 - A + 2.
\]


\paragraph{Potential Key-Collision Attacks with $m_2$}
\diff{
Assume the operating system assigns the victim process a secret key $k_v$ and the attacker process a key $k_a$. A potential concern is whether an attacker could exploit the $m_2$ modifier to induce a key collision across protection domains.
Specifically, suppose the attacker attempts to construct a relation of the form
$k_a = k_v \oplus m_{2,d}$ where
$m_{2,d}$ denotes a chosen modifier difference. If such a relation were achievable, the attacker could attempt to craft a tuple
$(m_{2,a} = m_{2,d} \oplus m_{2,v}, \; p_a = p_v, \; c_a = \mathsf{Enc}_{k_a}(p_v)),$
and inject it into the victim domain, thereby forging a valid pointer.
We prevent such a collision with our key assignment.}

\paragraph{\diff{Key 
Assignment}}
\diff{
To prevent cross-domain key-collision attacks induced by adversarial choices of $m_2$, the system enforces that the \emph{unaffected} portion of the domain key is unique across security domains. Concretely, if two domains share the same $128-|m_2|$ unaffected key bits, then an attacker could choose $m_2$ values that cause their derived key to match the victim's derived key. Ensuring uniqueness of the unaffected portion eliminates this possibility.
This policy also bounds the number of simultaneously supported security domains: since only $128-|m_2|$ bits are reserved for domain separation, the maximum number of unique domains is at most $2^{128-|m_2|}$. 
}

\diff{\paragraph{Proof of Security}
~

We define $\mathrm{Adv}^{}_{LIPPEN}(\mathcal{A})$ denote the advantage of adversary~$\mathcal{A}$ in forging an encrypted pointer $c_a$ for chosen $(m_a, p_a)$ for LIPPEN knowing a set of tuples $\mathcal{T}$, and $\mathrm{Adv}^{}_{E}(\mathcal{B})$ being
    the advantage of adversary $\mathcal{B}$ in distinguishing
    the block cipher $E_K(\cdot)$ from a uniformly random permutation $P$ even with oracle access to both $EncE_K(\cdot)$ and $DecE_K(\cdot)$.
For \sys, the implementation ensures that $|m_1|$ uses only unused pointer bits, and the $128-|m_2|$ key bits are never shared between domains.

\begin{theorem}
For any probabilistic polynomial-time (PPT) adversary $\mathcal{A}$ knowing $\mathcal{T}$ running in time $t$, there exists a PPT adversary $\mathcal{B}$ such that

\[
\mathrm{Adv}^{}_{\sys}(\mathcal{A})
\;\le\;
\mathrm{Adv}^{}_{E}(\mathcal{B})
\;+\;
\varepsilon(q).
\]

\end{theorem}

\begin{proof}
$\mathcal{B}$ is given oracle access to either a real block cipher $E(\cdot)$ or a random permutation $P$ for both encryption and decryption. $
\mathcal{B}$ runs $\mathcal{A}$ by simulating the protocol using its oracle in place of $E(\cdot)$ as follows.

\begin{enumerate}
  \item Answer all of $\mathcal{A}$'s queries to build $\mathcal{T}$ by querying the oracle, i.e., encryption or decryption to obtain $c_i$ with key $k_v\oplus m_{2,i}$ and plaintext $p_i\oplus m_{1,i}$.
  \item For $(m_a, p_a)$, run $\mathcal{A}$ to forge pointer $c_a$. $\mathcal{B}$ queries the oracle to encrypt $plain = p_a\oplus m_{1,a}$ with key $k_v\oplus m_{2,a}$. $\mathcal{B}$ repeats for $q$ times and if more than $q/2$ encryption of $plain$ matches $c_a$, $\mathcal{B}$ returns that a real block cipher $E(\cdot)$ is behind the oracle; otherwise, $\mathcal{B}$ returns that a random permutation $P$ is behind the oracle. 
\end{enumerate}
If the oracle is for a real block cipher, then the probability of $\mathcal{B}$ returning the right value (i.e., $\mathrm{Adv}^{}_{E}(\mathcal{B})$) is the probability $\mathcal{A}$ returning the right $c_a$ more than half of the time (i.e., $\mathrm{Adv}^{}_{\sys}(\mathcal{A})$). If the oracle is for a random permutation $P$, then $\mathcal{B}$ will return the wrong result if $c_a$ happens to be the output of the random permutation, which is of probability $2^{-64\times q/2}$. Thus, $\mathrm{Adv}^{}_{E}(\mathcal{B}) \geq \mathrm{Adv}^{}_{LIPPEN}(\mathcal{A})-\varepsilon(q) $
\end{proof}

Thus, breaking \sys is not easier than breaking the underlying block cipher. 
The security of the scheme therefore reduces to the cryptographic strength of the underlying PRINCE-family cipher. Existing cryptanalysis of the PRINCE family primarily targets reduced-round variants or relies on implementation attacks such as differential fault analysis~\cite{jean2014security_prince, song2013dfa_prince,soleimany2015reflection,cryptoeprint:2012/712,10.1007/s00145-020-09345-0,cryptoeprint:2015/245} and has not produced practical attacks on the full-round constructions, indicating no practical key-recovery; the best known attacks require significantly reduced-round variants or complexity close to exhaustive search. More specifically, the best cryptanalysis efforts in PRINCE (base design for PRINCEv2) report that a single key can be recovered with a computational complexity of $2^{125.47}$ using structural linear relations; in the related key setting, the memory complexity is $2^{33}$ and the time complexity $2^{64}$; using the related key boomerang attack, the complexity is $2^{39}$ for both memory and time~\cite{jean2014security_prince}. The authors of PRINCEv2 claim that there is no attack against PRINCEv2 with memory complexity below $2^{47}$ (chosen) plaintext-ciphertext pairs (obtained under the same key) and time-complexity below $2^{112}$~\cite{princev2}, which is in line with the NIST requirement on the security of lightweight ciphers~\cite{nist2018_lwsecurity}. }

\subsubsection{\diff{Integrity with Encryption}}
\diff{
Theorem 1 shows that the probability of an adversary successfully forging a pointer for a target address and modifier is no larger than breaking the cipher. For a random encrypted pointer, after decryption, if the unused pointer bits do not match $m_1$, the engine will detect that the encrypted pointer is not valid  with detection probability of $1-2^{-|m_1|}$, which is at the same security level as PAC. Even if the modifier matches, the pointer will point to a random place in the address space not controlled by the attacker, while with PAC the attacker can forge a pointer to an attacker-chosen address.
}

\subsubsection{Number of Modifier Bits Needed in Practice}
\label{sec:num_mod}

Given that the security will depend on the number of bits in the modifier, here we study the entropy needed in the modifier in practice.

\begin{table}[t]
\centering
\caption{Modifier in representative PAC-based defenses.}
\label{tab:modifier_usage}
\renewcommand{\arraystretch}{1.1}
\setlength{\tabcolsep}{4pt}
\begin{tabular}{|l|p{0.50\columnwidth}|}
\hline
\textbf{Work} & \textbf{Modifier Design} \\ \hline

\textbf{PARTS}~\cite{liljestrand2019pac} &
Stack pointer (SP) for return addresses, and a type identifier for indirect and data pointers. \\ \hline

\textbf{PACStack}~\cite{PACStack} &
Previous return address on the stack. \\ \hline

\textbf{PTAuth}~\cite{farkhani2021ptauth} &
A generated object-id. \\ \hline

\textbf{PACSan / PACMem}~\cite{li2022pacmem, li2022pacsan} &
A static random number generated at compile time. \\ \hline

\textbf{AOS}~\cite{kim2020hardware} &
Stack pointer (SP) for return addresses. \\ \hline

\textbf{PACTight}~\cite{ismail2022tightly} &
Pointer location and a random tag for sensitive data pointers; previous return address and unique function ID for return addresses. \\ \hline

\textbf{RSTI}~\cite{ismail2024enforcing} &
Unique mixture of pointer scope, type, permission, and location information. \\ \hline

\end{tabular}
\end{table}

Different pointer-protection schemes employ different modifier types, as summarized in Table~\ref{tab:modifier_usage}. 
Apple's PAC~\cite{apple:pac} uses a zero modifier for function and vtable pointers, requiring only a single unique value
In general, the number of modifier bits required is proportional to the number of distinct modifier values that must be represented.

PARTS-CFI~\cite{liljestrand2019pac} uses a 64-bit type identifier, but the effective entropy needed depends on the number of distinct pointer types and variables in the program (e.g., \texttt{int*}, \texttt{char*}, etc.) 
For example, in \texttt{xalancbmk}, one of the largest benchmarks in both SPEC CPU2006 and SPEC CPU2017, there are 2,558 pointer types and 32,097 pointer variables~\cite{ismail2024enforcing}, corresponding to roughly 12 bits of modifier space to ensure uniqueness.
The same reasoning extends to other PAC-based defenses. 
RSTI~\cite{ismail2024enforcing} derives modifiers from pointer scope, type, permission, and location, increasing the unique context count to 14,073 for \texttt{xalancbmk} and requiring approximately 14 bits; assigning a distinct modifier per pointer variable would require 16 bits. Since this benchmark represents the largest pointer footprint across the SPEC suites, we conclude that a 16-bit modifier space is sufficient for realistic workloads.

\diff{Since the use of $m_2$ reduces the maximum number of unique domains, once we decide the $|m|$ based on the required entropy for context, we use all $|m_1|_{max}$ for $m_1$, and the remaining entropy is assigned to $m_2$.}

\subsection{\diff{Memory Tagging and Address Width Scaling}}

\diff{
Modern architectural trends, such as Memory Tagging Extensions (MTE) and the expansion of Virtual Address (VA) widths, significantly constrain the available non-canonical bits within a 64-bit pointer. Memory Tagging associates memory regions with metadata tags stored in the upper pointer bits to detect spatial and temporal violations. Simultaneously, scaling the address width directly reduces the unused bits previously available for in-pointer security metadata.

\sys is designed to be agnostic to these architectural shifts. Since our encryption operates transparently on the pointer value, it preserves the integrity of any bits reserved by hardware for addressing or tagging. However, as the address space $A$ grows or the tag field $|tags|$ expands, the bit-budget for modifier $m_1$ is proportionally reduced. To maintain a constant security margin, \sys can meet the entropy requirements by using secondary modifier $m_2$.
}

\diff{
Under these constraints, the number of supported distinct security domains is bounded by the remaining key-separation entropy. Specifically, the effective domain-separation space becomes
$E = 128 - |m| - |Tag| + (64 - A)$. For example, with $|m|=16$ and $|Tag|=4$:  
(i) if $A=48$, then $E=124$ bits ($2^{124}$ domains);  
(ii) if $A=57$ (x86\_64 servers with 5-level page tables), then $E=115$ bits ($2^{115}$ domains).
}

\subsection{Discussion on Speculative Execution and Performance}

Pointer encryption serializes decryption with pointer dereference, introducing non-zero latency overhead on every protected access. PARTS~\cite{liljestrand2019pac} quantifies this at 4 cycles per dereference, yielding $<0.5\%$ overhead for code pointers but $\sim$20\% for all data pointers in nbench. We corroborate this effect using a pointer-chasing microbenchmark that measures the cost of accessing signed pointers on the Apple M1 processor.
Although such overheads may be acceptable for PAC-style deployments, they motivate architectural optimizations to reduce dereference latency.
Prior designs like C3~\cite{lemay2021cryptographic} discuss the optimizations like predictions, showing close to zero protection overhead. But C3 design is based on Intel architectures, limiting some optimization.

Code pointers are dereferenced in branch, jump, or return instructions. 
The branch predictor will still work as is. Branch target prediction like branch target buffer (BTB) and return \diff{address stack (RAS)} usually uses the PC of the current branch instruction for prediction instead of the pointer itself. Pointer encryption does not touch the predictor design, and thus, the prediction rate of the branch target will not be affected. With pointer encryption, the resolution of the branch will take one more cycle, which adds to the execution latency when a branch misprediction happens and has negligible overhead on a correct branch prediction. With a decent branch prediction, the performance overhead will be small, as shown in designs like PAC. \diff{We provide further proof of the effects of BTB and RAS on overhead in section~\ref{sec:evaluation}.}

%% file: 5.implementation.tex
\section{Implementation}

\paragraph{Hardware}We prototype \sys on a 64-bit RISC-V platform implemented on an AMD/Xilinx VCU118 (Virtex UltraScale+) FPGA \diff{booting a FireMarshal-managed Linux image~\cite{pemberton2021firemarshal}}. 
The system is built using Chipyard~\cite{amid2020chipyard} v1.8 and synthesized with Xilinx Vivado 2021.2. \diff{Our design extends both the in-order Rocket and out-of-order BOOM cores~\cite{asanovic2016rocket} }with a custom cryptographic accelerator for PRINCEv2 connected via the Rocket Custom Coprocessor (RoCC) interface. The accelerator operates on 64-bit pointers and communicates with the core through tightly coupled request and response queues. \diff{ We also implement QARMA in RoCC on the FPGA platform to use as an authentication-based PAC baseline.
We configure \sys with \(M_1 = 16\) bits and \(M_2 = 0\), as 16 bits suffice to uniquely distinguish all pointer contexts. 
Our Rocket and Large Boom cores are configured with default settings.
}

\paragraph{Compiler Support}

\diff{For return address protection,}  we use the stack pointer as the modifier input, which is similar to \texttt{RETAA} for PAC, and our compiler support is implemented in LLVM~\cite{lattner2004llvm} v18.1 by extending the RISC-V backend to instrument call and return sites with \sys sealing and unsealing instructions. 

\diff{To demonstrate compatibility, we leverage PacTight~\cite{ismail2022tightly} and its LLVM-based compiler framework, making minor modifications to its IR pass to target the RISC-V architecture.}

%% file: 6.evaluation.tex
\section{Evaluation}
\label{sec:evaluation}
\diff{
\subsection{Evaluation Method}

\textbf{Evaluation Setup.} 
We evaluate on both Rocket (in-order) and BOOM (out-of-order) cores on our FPGA platform running at 100,MHz. \sys provides pointer protection via PRINCEv2-based encryption, while our PAC implementation does authentication using our QARMA implementation; we additionally compare against Apple's PAC on the commercial M1 processor.}

\textbf{Performance Overhead.} 
We evaluate runtime overhead and the number of instructions retired relative to uninstrumented baselines across our benchmark suites. These metrics capture the execution cost introduced by \diff{the added instructions (e.g., fetching, decoding in the pipeline) and} pointer sealing and unsealing operations in \sys, reflecting both time and instruction-level overheads. 
\diff{Each experiment is repeated at least twice to ensure it is not affected by significant noise.}

\diff{\textbf{Compatibility.}
To demonstrate compatibility with prior PAC-based protection in the compiler, we adopt PacTight~\cite{ismail2022tightly} from Arm PAC to LIPPEN in RISC-V, to evaluate the migration complexity.}

\textbf{Hardware Cost.}
We report FPGA resource utilization (LUTs and flip-flops) for our Rocket and BOOM core with and without the encryption accelerator. This metric evaluates the hardware footprint of the pointer encryption logic.

\textbf{Power Consumption.}
Power estimates are obtained from post-synthesis analysis on the VCU118 board. This metric complements area evaluation and quantifies the energy efficiency of the encryption accelerator.

\begin{figure*}[t]
  \centering
  \includegraphics[width=\textwidth]{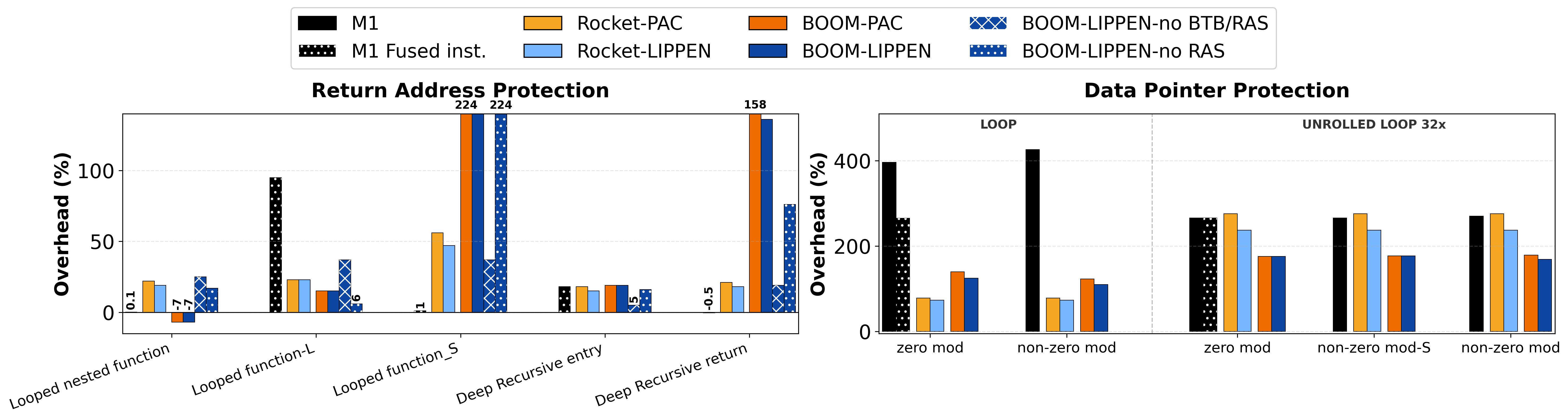}
  \caption{\diff{Microbenchmark results for return-address (left) and data-pointer (right) protection. Runtime slowdowns are normalized to the unprotected baseline on the same processor and configuration.}}
  \label{fig:micro_bench_result}
\end{figure*}

\subsection{Benchmarks}
\diff{
\textbf{Micro-benchmarks.}
We design a set of targeted microbenchmarks to isolate the runtime overheads of data pointer and return address protection. 

\textit{Return address protection.}
To isolate signing and authentication costs under different call-stack behaviors, we design four microbenchmarks:
(1) Looped nested function calls: a loop of nested function calls with a depth of 8 to represent common user programs with nested function calls.
(2) Looped function calls: a function repeatedly invoked in a tight loop. We implement two variants: \texttt{function\_S}, containing a single \texttt{xor}, and \texttt{function\_L}, containing three \texttt{xor}s and one memory access.
Timing measurements include both signing and authentication, capturing the combined per-call overhead.
(3) Deep Recursive entry-only: a very deep recursive function{ (depth=4096)} that measures only function entry (and signing).
(4) Deep Recursive return-only: a very deep recursive function (depth=4096) that measures only function return (and authentication).
All benchmarks are automatically instrumented: on RISC-V using our LLVM-based compiler pass, and on Arm using the \texttt{arm64e} compilation flag. 
Together, these microbenchmarks measure the performance of function entry and return in different return address prediction scenarios, enabling precise characterization of \sys’s protection overhead.

\textit{Data pointer protection.}
We implement a pointer-chasing benchmark that forms a strictly serialized load chain, ensuring authentication lies on the critical path. Two variants are evaluated: (1) \emph{Loop}, with a single dependent load per iteration, and (2) \emph{Unrolled}, with 32 dependent loads per iteration to amortize the impact of loop instructions. For each, we test authentication with (a) zero modifier, (b) a shared non-zero modifier for all accesses per iteration, and (c) a non-zero modifier loaded per access. 
We manually insert load and authentication instructions at the assembly level to ensure precise control over the execution sequence.
We additionally evaluate Apple M1’s fused load-and-authenticate instruction. These configurations isolate intrinsic authentication latency, modifier-fetch overhead, and fusion benefits.

}

\diff{\textbf{End-to-end evaluation.}}

\diff{
We first evaluate \sys using the nbench suite to characterize overhead on lightweight compute kernels. These benchmarks stress arithmetic and memory subsystems in isolation, enabling controlled measurement of pointer-protection cost without full-application complexity.}

We then evaluate \sys on the SPEC CPU2017 rate suite, covering diverse compute- and memory-intensive workloads representative of modern systems. Our study includes C and C++ benchmarks compiled with LLVM-based instrumentation, spanning integer workloads (e.g., \texttt{perlbench\_r}, \texttt{gcc\_r}, \texttt{xalancbmk\_r}) and floating-point workloads (e.g., \texttt{lbm\_r}, \texttt{namd\_r}), thereby capturing varied control-flow and memory-access behaviors.
\diff{Certain SPEC benchmarks are excluded due to interactions with the C++ exception-handling runtime. During stack unwinding, sealed pointers may be dereferenced without prior unsealing, causing incorrect control flow. Supporting this would require modifications to system libraries (e.g., recompiling GCC runtime components to insert unsealing), which is orthogonal to \sys’s architectural design and left to future work.}

\subsection{Performance Results}
\label{sec:performance-results}
\diff{\textbf{Micro-benchmark results.}

\diff{
\textit{(a) Return address protection.}
The left side of Figure~\ref{fig:micro_bench_result} reports the overhead of return-address protection.
Overall, \sys on BOOM is comparable to Apple’s M1 in most cases, and \sys has a slightly smaller overhead than PAC on the FPGA platforms.
The dominant factor influencing performance is return prediction rather than the protection primitive itself. 
For example, in the looped nested-function call benchmark, both BOOM and Apple M1 show negligible overhead or even small performance improvements.
 Disabling the RAS on BOOM significantly increases overhead. 
For looped single function calls, the overhead varies across architectures and function body sizes.  The extremely small function body (looped function-S) increases misprediction frequency on BOOM, leading to substantially higher overhead compared to looped function-L, while M1 has the opposite behavior. 
 For deep recursive returns, where the BOOM RAS capacity is exceeded, disabling both the RAS and BTB on BOOM can reduce overhead, as repeated mispredictions and recovery penalties otherwise amplify authentication latency. Meanwhile, Apple M1 shows negligible overhead because M1's predictor can still make correct return predictions. Deep recursion entry operations, however, show relatively low sensitivity to prediction structures, since signing occurs prior to control-flow resolution and Rocket, BOOM, and M1 show similar overheads.
When prediction mechanisms are effective, BOOM achieves overhead comparable to commercial M1 implementations while preserving the protection guarantees of \sys.

\textit{(b) Data pointer protection.} 
The right side of Figure~\ref{fig:micro_bench_result} reports the overhead of data-pointer protection due to authentication.
 The overhead of \sys in BOOM is comparable to that of Apple’s M1. On FPGA,  \sys consistently tracks QARMA-based PAC implementations. 
Interestingly, introducing an additional load for the modifier does not significantly affect performance across architectures. Although a non-zero modifier requires additional load to fetch its value, this load is typically independent of the critical load-use chain and can overlap with other in-flight operations and in our benchmark the load of modifier will always hit in L1. 
Finally, we attribute the difference between looped and unrolled 32$\times$ configurations to the differences in how the pipelines handle loops, e.g., branch prediction. The unrolled 32$\times$ amortized the effect of loops. 
These results demonstrate that full pointer encryption is not more expensive than pointer authentication in commercial processors while providing a much stronger security level.
}

\begin{figure*}[t]
  \centering
  \includegraphics[width=\textwidth]{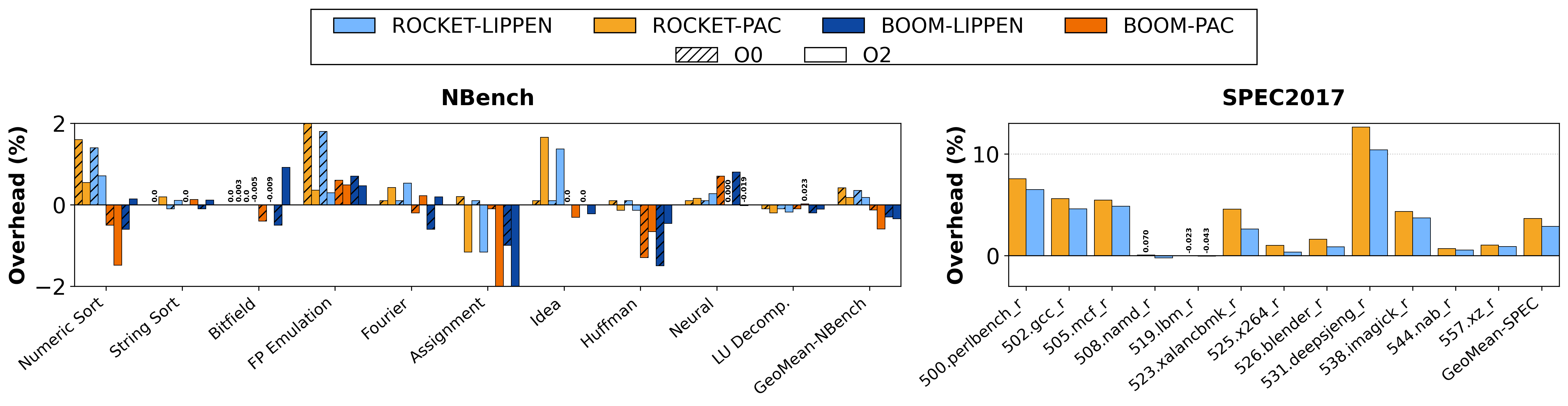}
  \caption{\diff{Overhead across NBench (left) and SPEC CPU2017 (right). Purple bars indicate dynamic instruction count increase, while other bars show runtime overhead, all normalized to the corresponding unprotected baseline on the same processor.}}
  \label{fig:overheads}
\end{figure*}

}

\diff{
\textbf{End-to-end evaluation results.}

Figure~\ref{fig:overheads} reports the runtime overhead of \sys relative to PAC across nbench and SPEC2017, normalized to the uninstrumented baseline. 
On nbench, both \sys and PAC incur only 0.2\% additional dynamic instructions and a 0.2\% geometric mean overhead on Rocket under O2 optimization, while BOOM overhead is effectively 0\%. Under O0, overheads remain comparably low, with geometric means of 0.35\% and 0.42\% for \sys and PAC on Rocket, and near-zero on BOOM. These results are consistent with prior reports for return-address protection on nbench (e.g., 0.5\% in PARTS~\cite{liljestrand2019pac} and 0.11\% in Rettag~\cite{wang2022rettag}), confirming negligible cost when call density is low.
Across SPEC on Rocket, return-address protection increases dynamic instructions by ~1\% on average, yielding geometric mean overheads of 2.9\% for \sys and 3.6\% for PAC. Control-flow-intensive workloads (e.g., \texttt{perlbench\_r}, \texttt{leela\_r}, \texttt{deepsjeng\_r}) exhibit higher overheads (6--12\%), whereas compute- and memory-bound applications (e.g., \texttt{namd\_r}, \texttt{lbm\_r}, \texttt{nab\_r}) remain near baseline.
Overall, \sys matches, and occasionally slightly outperforms PAC while providing stronger integrity and confidentiality guarantees. Overheads remain modest and are driven primarily by control-flow intensity rather than by full-pointer encryption.

}

\diff{
\subsection{Compatibility with Prior Work}

\begin{figure}[t]
  \centering
  \includegraphics[width=\columnwidth]{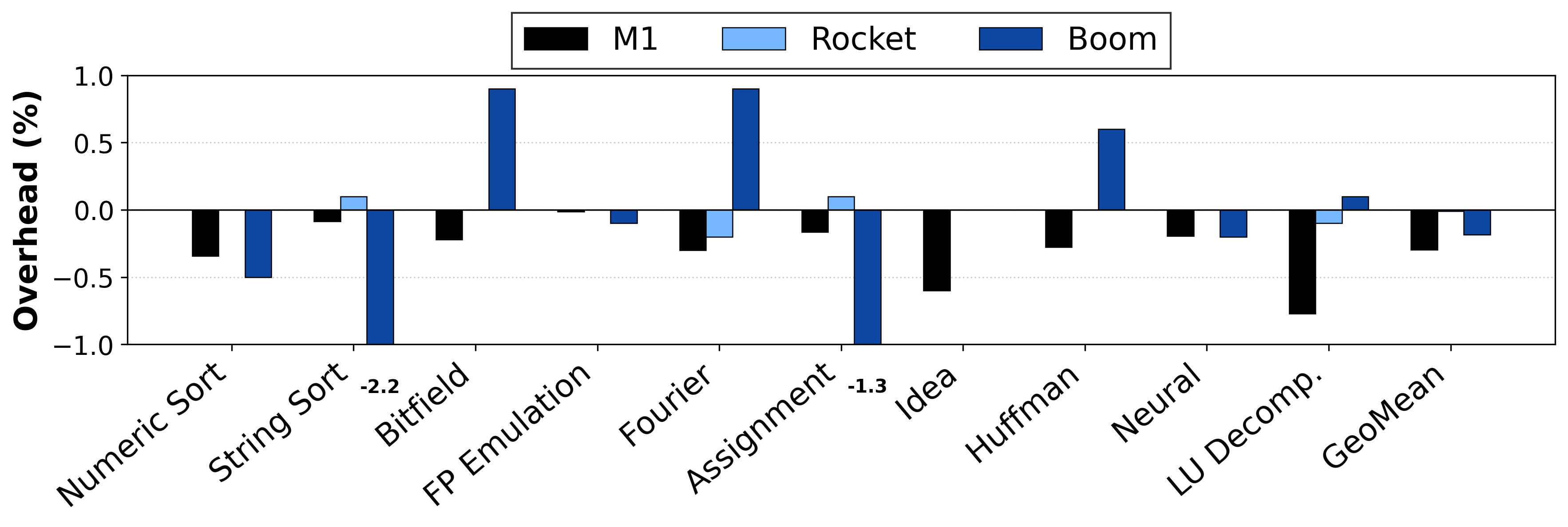}
  \caption{\diff{Performance comparison of PacTight on different cores}}
  \label{fig:pactight}
\end{figure}

PacTight~\cite{ismail2022tightly} enforces pointer integrity using strong, unique modifiers to protect sensitive pointers and provide spatial and temporal memory safety. It instruments programs via an LLVM IR pass that inserts PAC signing and authentication instructions.
To evaluate compatibility, we integrated \sys into PacTight’s compilation flow by replacing the original Arm \texttt{pacia} and \texttt{authia} instructions with our RISC-V \texttt{seal} and \texttt{unseal} instructions. This required only minor modifications to the IR pass (fewer than 50 lines of code), indicating minimal integration effort and no structural changes to the underlying protection model.
While small ISA-specific differences exist in the IR representation between Arm and RISC-V, the resulting binaries are instrumented equivalently to the original PacTight design (Complete similarity in the PacTight's provided example code and nbench). 
As presented in Figure~\ref{fig:pactight}, compiling and running nbench with the adapted PacTight infrastructure on Apple M1, Rocket, and BOOM results in negligible performance overhead, typically under 1\%. Notably, the Apple M1 and BOOM consistently exhibit slight performance improvements (negative overhead). These results demonstrate that \sys composes cleanly with prior PAC-based defenses while preserving their low cost profile across different microarchitectures.

As presented in Figure~\ref{fig:pactight}, running nbench with the adapted PacTight infrastructure on Apple M1, Rocket, and BOOM incurs negligible overhead under \texttt{-O0}, demonstrating that \sys composes cleanly with prior PAC-based defenses while preserving their expected behavior and cost profile. Under \texttt{-O2}, M1 overhead remains near-zero, while Rocket and BOOM increase to $7.50\%$ and $2.88\%$, respectively. We attribute this to the compiler treating \sys instructions as opaque barriers, preventing code motion and scheduling optimizations around them. 
}

\begin{table}[t]
  \centering
  \caption{FPGA area and power comparison \diff{on Rocket}.}
  \label{tab:area-overhead}
  \renewcommand{\arraystretch}{1.1}
  \setlength{\tabcolsep}{1.2pt} 
  \hspace{-6pt}
  \begin{tabular*}{\columnwidth}{@{\extracolsep{\fill}} l cccccc cc @{}}
    \toprule
    \textbf{Config} &
    \multicolumn{2}{c}{\textbf{Chip}} &
    \multicolumn{2}{c}{\textbf{Core}} &
    \multicolumn{2}{c}{\textbf{ROCC}} &
    \textbf{Max Freq.} &
    \textbf{Power} \\
    \cmidrule(lr){2-3} \cmidrule(lr){4-5} \cmidrule(lr){6-7}
     & \textbf{LUT} & \textbf{FF} 
     & \textbf{LUT} & \textbf{FF} 
     & \textbf{LUT} & \textbf{FF} 
     & \textbf{(MHz)} & \textbf{(W)} \\
    \midrule
    Rocket-base     & 56{,}311 & 41{,}856 & 28{,}942 & 14{,}570 & --- & --- & 150 & 3.935 \\
    Rocket-RoCC    & 56{,}498 & 41{,}911 & 29{,}120 & 14{,}607 & 47 & 3 & 110 & 3.934 \\
    Rocket-\sys & 58{,}137 & 42{,}100 & 30{,}837 & 14{,}793 & 1{,}034 & 131 & 99 & 4.035 \\
    Rocket-PAC    & 58{,}519 & 42{,}110 & 31{,}212 & 14{,}791 & 2{,}071 & 132 & 89 & 4.031 \\
    BOOM-base    & 241{,}340 & 97{,}455 & 227{,}187 & 91{,}609 & --- & --- & 93 & 5.389 \\
    BOOM-RoCC    & 248{,}914 & 98{,}245 & 234{,}783 & 92{,}392 & 48 & 3 & 86.5 & 5.532 \\
    BOOM-\sys    & 248{,}416 & 98{,}674 & 234{,}289 & 92{,}819 & 862 & 193 & 90.6 & 5.625 \\
    BOOM-PAC    & 251{,}299 & 97{,}919 & 236{,}953 & 92{,}074 & 2{,}054 & 131 & 73.5 & 5.606 \\

    \bottomrule
  \end{tabular*}
\end{table}

\subsection{Hardware Cost and Power Results}

Table~\ref{tab:area-overhead} summarizes FPGA resource utilization, maximum frequency, and power (measured at 100,MHz) across all configurations on both cores, where Rocket-RoCC and BOOM-RoCC instantiate the ROCC interface without any cipher, Rocket-\sys and BOOM-\sys add the PRINCEv2-based encryption accelerator, and Rocket-PAC and BOOM-PAC add the QARMA-based authentication accelerator. For both cores, the QARMA configuration incurs slightly more logic than PRINCEv2 due to its more complex cipher datapath. The RoCC-only variants reveal that the ROCC interface itself accounts for the majority of the frequency reduction, with the cipher contributing only marginally on top for Rocket; on BOOM, however, the cipher datapath plays a more significant role in degrading the maximum frequency. Power overhead remains marginal across all configurations on both cores. Overall, \sys achieves full-pointer encryption at hardware and power costs commensurate with cipher complexity, confirming its practicality for integration into modern processor pipelines. 

%% file: 6.discussion.tex
\section{\diff{Discussion}}
\subsection{\diff{Secure Key Management}}
\diff{

Secure key management is essential to prevent key exposure and cross-domain forgery. While Arm Pointer Authentication (PA) provides hardware key registers~\cite{armv8_arm_key}, sharing keys across exception levels can enable cross-domain attacks~\cite{cai2023demystifying}. Recent designs, such as Apple’s M-series processors, address this via per-VM, per-EL, and per-boot key isolation using hardware-backed diversification and internal key derivation~\cite{cai2023demystifying}.
\sys can use the same mechanisms. Because sealing and unsealing are keyed primitives analogous to PAC instructions, domain-specific keys and hierarchical derivation apply directly. Thus, PAC-style key isolation extends to \sys without architectural changes, keeping key management orthogonal to the encryption mechanism while leveraging existing secure hardware deployments.

}

\subsection{\diff{Pointer Arithmetic}}
\diff{
A potential concern when protecting pointers is the cost of pointer arithmetic. In C and C++, pointers can be incremented or adjusted (e.g., during array traversal), and pointer authentication and encryption will complicate the arithmetic by requiring an authentication/decryption before the pointer arithmetic. However, pointer arithmetic without a subsequent memory access is rare in practice. Most arithmetic operations on pointers are immediately followed by pointer dereference, and one can optimize the order of decryption and arithmetic (in the compiler). Implementations of prior PAC-based systems that protect all pointers, such as RSTI~\cite{ismail2024enforcing} and AOS~\cite{kim2020hardware}, show that pointer arithmetic is not the main bottleneck, and the main bottleneck is still the authentication (decryption) itself. 
}

%% file: 7.related.tex
\section{Related Work}
We have so far discussed mitigation techniques that ensure pointer integrity, focusing primarily on approaches with zero memory footprint and straightforward deployability on existing hardware. In this section, we broaden the discussion to include other methods that aim to provide general memory safety.

\subsection{Capability and Tagged Architectures}

Capability-based and tagged architectures associate each pointer with hardware-managed metadata encoding bounds, permissions, and provenance, validated on every memory access. CHERI~\cite{woodruff2014cheri} is the most prominent example, providing strong spatial and temporal memory safety at the cost of expanded pointer representations and significant ISA and ABI changes that complicate deployment within existing software ecosystems.

Secure tagged architectures~\cite{gollapudi2023control} similarly attach metadata tags to memory and registers to enforce security policies at fine granularity. These tags can represent pointer authenticity, data provenance, or privilege levels, and are checked dynamically during instruction execution. 
ZeRØ~\cite{ziad2021zero} introduces new memory instructions and a metadata \emph{encoding} scheme (e.g., type bits in pointers plus per-line bit-vectors) so that pointer locations can only be written by authorized instructions; invalid writes are rejected, allowing resilient operation under attack. 
Tag-based enforcement provides flexibility but incurs additional storage and lookup overheads and often depends on compiler or OS support to manage tag propagation.

\subsection{Memory Safety}

A complementary class of defenses targets memory safety rather than pointer integrity~\cite{devietti2008hardbound, burow2018cup, sasaki2019practical, serebryany2018memory, lin2024camp, sparcadi, armmte, schrammel2023memes, ziad2021no, farkhani2021ptauth, li2022pacmem, kim2020hardware}. These mechanisms associate additional metadata with memory regions to detect spatial and temporal errors such as buffer overflows and use-after-free vulnerabilities. Hardware memory tagging schemes, such as Arm’s Memory Tagging Extension (MTE)~\cite{armmte} and SPARC ADI~\cite{sparcadi}, assign small tags (typically 4--8 bits) to memory blocks and to pointers. Each memory access compares the pointer’s tag with the memory’s tag, triggering an exception on mismatch. Software-based approaches~\cite{serebryany2018memory, lin2024camp} implement a similar mechanism in software with compiler instrumentation.

A recent line of work leverages pointer authentication for enforcing spatial and temporal memory safety~\cite{farkhani2021ptauth, li2022pacmem, kim2020hardware}. They build upon Arm Pointer Authentication Codes (PAC) to cryptographically bind pointers to their allocation or lifetime metadata, thereby preventing out-of-bounds and use-after-free violations. 

Other proposals extend tagging granularity or functionality. Califorms~\cite{sasaki2019practical} introduces fine-grained, byte-level tagging to blacklist invalid memory regions, while MEMES~\cite{schrammel2023memes} employs lightweight encryption to enforce spatial and temporal safety on commodity hardware. Similarly, systems such as CUP~\cite{burow2018cup}, HardBound~\cite{devietti2008hardbound}, CAMP~\cite{lin2024camp}, and No-FAT~\cite{ziad2021no} combine compiler or architectural support with metadata tracking to enforce object and bounds checking. These designs strengthen full memory safety but maintain separate metadata storage or architectural state, resulting in additional memory footprint and runtime overhead.

%% file: 8.conclusion.tex
\section{Conclusion}
This paper presented \sys, a low-overhead pointer-encryption architecture that eliminates the limitations of authentication-based pointer integrity enforcement through full-pointer encryption. By encrypting the entire pointer value rather than storing a truncated authentication code, \sys maximizes entropy, achieves brute-force resilience, and unifies protection for both code and data pointers. The design maintains seamless compatibility with the existing PAC software stack and ABI, requiring no changes to compilers or operating systems. 
Our prototype on a RISC-V FPGA platform protects return addresses and demonstrates that \sys delivers strong pointer integrity with lower cost compared to PAC and minimal hardware and power increases (below 4\%). These results show that full-pointer encryption is both practical and effective, offering cryptographic-level protection against pointer forgery while remaining suitable for deployment in modern processors.

\section{Acknowledgment}
At Virginia Tech, this project is partially supported by Commonwealth Cybersecurity Initiative (CCI), by the National Science Foundation (NSF) under grant CCF-2153748, CNS-2442993. LLMs were used for editorial purposes, with all outputs inspected by the authors to ensure accuracy and originality.

%% file: 9.appendix.tex
\section{Appendix}
\subsection{Abstract}
The artifact is designed to enable reproduction of the core components and a representative subset of the evaluation results, while supporting full reproduction given additional hardware and setup time.

The artifact guides users through compiling and simulating the Chipyard design with Verilator and running small test programs. It also includes instructions for generating a bitstream for the FPGA implementation. To run the microbenchmarks on our prototype, users need access to a Xilinx VCU118 FPGA and must follow the documented steps for building the Linux image, preparing the SD card with the required RISC-V binaries, and transferring files through the SD card workflow. In addition, the artifact includes microbenchmarks for ARM64 processors, which we tested on Apple M1 systems and expect to work on other Apple M-series processors as well. Because Apple restricts the use of PAC instructions in user-space programs, users must disable System Integrity Protection before running those experiments.

\subsection{Artifact check-list (meta-information)}
\begin{itemize}
    \item \textbf{Compilation:} Compiling the modified Chipyard hardware design, the LLVM-based compiler toolchain, and the provided microbenchmarks.
    \item \textbf{Run-time environment:} Ubuntu Linux for building and simulation; firemarshal linux for runnign on FPGA; macOS for ARM64 microbenchmark experiments.
    \item \textbf{Hardware:} A server or workstation for building and simulation; a Xilinx VCU118 FPGA for FPGA-based experiments; an Apple M1-based machine for ARM64 microbenchmark evaluation.
    \item \textbf{Execution:} Running Verilator-based simulations, generating FPGA bitstreams, booting Linux on the FPGA prototype, and executing the provided microbenchmarks both on FPGA and Apple M1.
    \item \textbf{Output:} Generated simulation binaries, FPGA bitstreams, compiled benchmark binaries, and performance measurement results.
    \item \textbf{How much time is needed to prepare workflow (approximately)?:} Around 30--60 minutes for software and simulation setup; the FPGA synthesis and bitstream generation take hours.
    \item \textbf{Publicly available?:} Yes, artifacts can be found here:
    https://doi.org/10.5281/zenodo.19901476 

    https://github.com/bearhw/LIPPEN
    \item \textbf{Code licenses (if publicly available)?:} GNU General Public License v3.0.
 
\end{itemize}

\subsection{Description}

\subsubsection{How to access}
The artifact can be accessed here https://doi.org/10.5281/zenodo.19901476 .  It can also be accessed through the public GitHub repository by cloning the repository https://github.com/bearhw/LIPPEN. The README file in the repository contains the full instructions for building and running the artifact.

\subsubsection{Hardware dependencies}
The artifact supports multiple workflows with different hardware requirements. For simulation and software compilation, a standard Linux server or workstation is sufficient. For FPGA-based evaluation, the artifact requires a Xilinx VCU118 FPGA board. For the ARM64 microbenchmark experiments, the artifact requires an Apple M1 system, and it should also work on other Apple M-series processors.

\subsubsection{Software dependencies}
The artifact requires a Linux environment with the dependencies needed to build the modified Chipyard design, the LLVM-based compiler toolchain, and the provided benchmarks. Running the FPGA workflow additionally requires the tools needed to generate the bitstream, build the Linux image, prepare the SD card, and execute binaries on the prototype. For the ARM64 microbenchmark workflow, the artifact requires macOS with the appropriate build tools installed. Because Apple restricts the use of PAC instructions in user-space programs, System Integrity Protection must be disabled before running those experiments.

\subsection{Installation}
Clone the GitHub repository. Then follow the instructions in the README file to set up the build environment, compile the required components, and prepare the selected workflow. The README describes separate steps for simulation, FPGA-based evaluation, and ARM64 microbenchmark experiments.

\subsection{Experiment workflow}
\begin{enumerate}
    \item Clone the repository and initialize the required submodules and dependencies.
    \item Build the modified Chipyard hardware design and the LLVM-based compiler toolchain by following the provided scripts.
    \item Run the Verilator-based simulation flow to verify the design and execute the included small test programs.
    \item If FPGA evaluation is desired, generate the FPGA bitstream and prepare the Linux image for the VCU118 platform.
    \item Load the required binaries and files onto the SD card and boot the system on the FPGA prototype.
    \item Execute the provided microbenchmarks and collect the performance results.
    \item For the ARM64 workflow, compile and run the microbenchmarks on an Apple M1 or another Apple M-series processor after disabling System Integrity Protection.
\end{enumerate}

\subsection{Evaluation and expected results}
Successful execution of the artifact is demonstrated by:
\begin{enumerate}
    \item Correctly building the modified Chipyard design and LLVM-based toolchain.
    \item Running the provided test programs in the Verilator simulation.
    \item Reproducing the experimental results from Figure~4 of the paper.
\end{enumerate}

The FPGA-based experiments require access to a Xilinx VCU118 board and involve significant synthesis and setup time. 
Similarly, the ARM64 microbenchmark experiments require Apple M-series hardware and additional system configuration (e.g., disabling System Integrity Protection).